\documentclass[12pt]{article}

\usepackage[utf8]{inputenc}
\usepackage{enumerate}
\usepackage{mathtools}
\usepackage{mathrsfs}
\usepackage{amsmath}

\usepackage{setspace}
\usepackage{graphicx}
\usepackage{float,natbib}
\usepackage{subfig}
\usepackage{makecell}
\usepackage[toc, page]{appendix}
\usepackage[export]{adjustbox}
\usepackage{geometry}
\usepackage{amsthm,amsmath,amssymb}
\usepackage{dutchcal,color}
\usepackage{amsthm}
\usepackage{amssymb}  
\usepackage{verbatim}
\usepackage{pgfplots}
\usepackage{multirow}
\usepackage{mathabx}
\usepackage{tikz}
\pgfplotsset{compat=1.18}
\usetikzlibrary{arrows,automata}
\usepackage{float}
\newcommand{\Height}{0.7cm}
\newcommand{\Width}{1.4cm}

\tikzset{Square/.style={
    inner sep=0pt,
    text width=\Width,
    minimum size=\Height,
    draw=black,
    fill=red!0,
    align=center
    }
}
\usetikzlibrary{calc, decorations.pathreplacing, calligraphy}
\usetikzlibrary{bayesnet}

\begin{document}
\pdfoutput=1
\def\btheta{\boldsymbol{\theta}}
\def\trans{^{\scriptscriptstyle \sf T}}
\def\bzero{\boldsymbol{0}}
\def\by{\mathbf{y}}
\def\bgamma{\boldsymbol{\gamma}}
\def\S{\boldsymbol{S}}
\def\bx{\mathbf{x}}
\def\m{\boldsymbol{m}}
\def\bphi{\boldsymbol{\phi}}
\def\c{\boldsymbol{c}}
\def\C{\boldsymbol{C}}
\def\Var{\mathbf{Var}}
\def\bu{\boldsymbol{u}}
\def\bs{\mathbf{s}}
\def\bl{\mathbf{l}}
\def\G{\boldsymbol{G}}
\def\e{\boldsymbol{e}}
\def\bw{\mathbf{w}}
\def\bz{\mathbf{z}}
\def\bY{\mathbf{Y}}
\def\bbeta{\boldsymbol{\beta}}
\def\bfeta{\boldsymbol{\eta}}
\def\bz{\mathbf{z}}
\def\bpsi{\boldsymbol\psi}
\def\balpha{\boldsymbol{\alpha}}
\def\romone{\uppercase\expandafter{\romannumeral1}}
\def\romtwo{\uppercase\expandafter{\romannumeral2}}
\def\romthree{\uppercase\expandafter{\romannumeral3}}
\def\bX{\mathbf{X}}
\def\bY{\mathbf{Y}}
\def\bxi{\boldsymbol{\xi}}

\newtheorem{assumption}{ASSUMPTION}
\newtheorem{lemma}{Lemma}
\newtheorem{corollary}{Corollary}
\newtheorem{theorem}{Theorem}
\newtheorem{cond}{Condition}
\newtheorem{remark}{Remark}

\title{Robust and Efficient Semi-supervised Learning \\ for Ising Model}
\author{Daiqing Wu$^{1}$, Molei Liu$^{2}$ \bigskip \\
\small 
$^1${Department of Mathematics, Sun Yat-sen University} \\
\small 
$^2${Department of Biostatistics, Columbia University Mailman School of Public Health} \\
}

\date{}
\thispagestyle{empty}

\maketitle
\begin{abstract}
\vspace{0.1in}
\noindent
    In biomedical studies, it is often desirable to characterize the interactive mode of multiple disease outcomes beyond their marginal risk. Ising model is one of the most popular choices serving for this purpose. Nevertheless, learning efficiency of Ising models can be impeded by the scarcity of accurate disease labels, which is a prominent problem in contemporary studies driven by electronic health records (EHR). Semi-supervised learning (SSL) leverages the large unlabeled sample with auxiliary EHR features to assist the learning with labeled data only and is a potential solution to this issue. In this paper, we develop a novel SSL method for efficient inference of Ising model. Our method first models the outcomes against the auxiliary features, then uses it to project the score function of the supervised estimator onto the EHR features, and incorporates the unlabeled sample to augment the supervised estimator for variance reduction without introducing bias. For the key step of conditional modeling, we propose strategies that can effectively leverage the auxiliary EHR information while maintaining moderate model complexity. In addition, we introduce approaches including intrinsic efficient updates and ensemble, to overcome the potential misspecification of the conditional model that may cause efficiency loss. Our method is justified by asymptotic theory and shown to outperform existing SSL methods through simulation studies. We also illustrate its utility in a real example about several key phenotypes related to frequent ICU admission on MIMIC-III data set.

\end{abstract}


\section{Introduction}

\subsection{Background}

In recent years, interactive modes and graphical models of multiple diseases has become a popular topic in biomedical studies \citep{halu2019multiplex,del2019disease,hong2021clinical}. By studying the diseases as non-isolated elements in a graphical model, the researchers are able to reveal the intrinsic relationship among diseases, beyond their marginal risk probabilities and correlations. The detected disease network structures not only provide novel insight and understanding to clinical studies, but also have a great potential to be used for improving clinical decision making from many aspects such as reducing adverse treatment effects and re-purposing old drugs for new diseases \citep{del2019disease}. 

Meanwhile, most existing biomedical network analyses \citep{hong2021clinical,li2022graph} have been based on electronic health record (EHR) as this systematized collection storing digital health information is conveniently accessible on large study cohorts. However, EHR outcomes such as diagnostic codes and clinical mentions are usually informative yet too noisy to represent the true conditions accurately, which could incur severe bias \citep{zhang2019high}. On the other hand, medical chart reviewing by experts is a common approach to obtain accurate (gold-standard) labels for the diseases. Nevertheless, the extensive human efforts and time required by this process make the resulted sets of gold labels not scalable enough to support complicated analyses like graphical modeling and inference. To address such a problem in this paper, we consider a semi-supervised learning (SSL) setup of graphical models, aiming at leveraging a large sample of informative but error-prone auxiliary EHR features to assist learning with the small sample of subjects with gold standard labels for the disease status.

\subsection{Problem Setup}\label{sec:problem_setup}

Let $\by=(y_1,\ldots,y_q)\trans$ denote a binary random vector of $q$ phenotypes. First, we introduce the Ising model for $\by$ of our interests, with its probability mass function formulated as
\begin{equation}
    P(\by|\bw;\btheta)=\frac{1}{C(\bw;\btheta)}\exp\left(\sum_j(\btheta\trans_{jj}\bw)y_j+\sum_{k>j}\theta_{jk}y_jy_k\right),
    \label{equ:Ising model}
\end{equation}
where $C(\bw;\btheta)$ is the partitioning function, $\btheta=(\btheta_{11}\trans,\theta_{12},\ldots,\theta_{(q-1)q},\btheta_{qq}\trans)\trans$ are the model parameters, and $\bw=(1,w_1,\ldots,w_d)\trans$ contains adjustment features for the marginal model of each phenotype, such as age and gender. The Ising model in (\ref{equ:Ising model}) is capable of revealing the conditional dependence among the phenotypes since it implies that
\begin{equation}
    \log\bigg(\frac{P(y_j=1|\by_{\setminus j,\bw},\bw;\btheta)}{1-P_{\btheta}(y_j=1|\by_{\setminus j,\bw},\bw;\btheta)}\bigg)=\btheta_{jj}\trans\bw+\sum_{k:k\neq j}\theta_{jk}y_k=\btheta_j\trans\by_{\setminus j,\bw},
    \label{equ:Ising_meaning}
\end{equation}
where $\btheta_j=(\theta_{j1},\ldots,\btheta_{jj}\trans,\ldots,\theta_{jq})\trans$ and $\by_{\setminus j,\bw}=(y_1,\ldots,y_{j-1},\bw\trans,y_{j+1},\ldots,y_q)\trans$. Equation (\ref{equ:Ising_meaning}) is the conditional logistic model of the phenotype $y_j$ derived under (\ref{equ:Ising model}), with $\theta_{jk}=\theta_{kj}$ encoding the dependence of $y_j$ and $y_k$ conditional on $\bw$ and all the other phenotypes. Therefore, estimating such coefficients $\theta_{jk}$'s is crucial for understanding the network structure of $\by$.

In addition, suppose there are also auxiliary EHR features $\bx=(x_1,\ldots,x_p)\trans$ that are informative to $\by$. We will introduce detailed model structures and assumptions of $\bx$ in Sections \ref{sec:Aug} and \ref{sec:PoS}. In principle, to assist the learning of the network structure involving the interaction among $\by$, $\bx$ needs to be predictive of $\by$ not only marginally on each $y_j$ but also jointly on each pair of $(y_j,y_k)$, or even the whole $\by$. Let $\bz=\{\bx,\bw\}$ denote the whole set of covariates.

We now introduce the SSL setup considered in this paper. Due to the aforementioned difficulty of ascertaining the true disease status, only a small random subset of the subjects are observed on $\by$, while all subjects have the observation of features $\bz$ from EHR. Thus, we denote the collected data set as
\[ 
    \mathscr{D}=\mathscr{L}\cup \mathscr{U},\quad\mathscr{L}=\{\by^i,\bz^i\}_{i=1}^n=\{\by^i,\bx^i,\bw^i\}_{i=1}^n,\quad\mathscr{U}=\{\bz^i\}_{i=n+1}^{n+N}=\{\bx^i,\bw^i\}_{i=n+1}^{n+N}, 
\] 
where $\mathscr{L}$ stands for the labeled data set with the sample size $n$, $\mathscr{U}$ for the unlabeled data with size $N$. Assume that $N\gg n$, and all subjects $i=1,2,\ldots,n+N$ are independent and have the same underlying distribution of $\{\by,\bx,\bw\}$, although we do not observe the phenotypes $\by$ for $i=n+1,\ldots,n+N$. The supervised learning (SL) to be introduced in Section \ref{sec:method:sup} estimates
$\btheta$ in model (\ref{equ:Ising model}) through the logistic regression simply using $\mathscr{L}$, ignoring the large set $\mathscr{U}$ with $\bx$ informative of $\by$. Such insufficient utilization of the data motivates us to consider SSL methods leveraging $\mathscr{U}$ in this paper.

\subsection{Related literature}

Ising models \citep{ising1924beitrag} and Markov random fields \citep[MRF]{kindermann1980markov} have been frequently studied and applied for statistical network analysis in a wide range of fields such as genetics \citep{wang2011learning}, microbiomics \citep{kurtz2015sparse}, biomedical study \citep{hong2021clinical}, and social science \citep{guo2015estimating}. In this era of big data, methodological developments have been made to improve the robustness and efficiency of statistical inference of the Ising model or MRF. For example, \cite{wang2011learning} developed a penalized logistic regression approach for high-dimensional Ising models accounting for spatial correlation among the binary responses (nodes). \cite{cheng2014sparse} proposed a sparse Ising model framework accommodating covariate-specific conditional dependency between the responses. \cite{yang2018semiparametric} proposed and studied a new class of semiparametric exponential family graphical models. Their models are free of any parametric assumption or type specification on each response marginally, and, thus, more flexible and robust than the previous methods for MRF. 

In the context of high-dimensional inference, \cite{xia2018multiple} developed a multiple testing approach to detect between pathway (community) interactions in a gene expression network. \cite{cai2019differential} studied the global and multiple testing on the differential network between two Ising models. Nevertheless, all these existing approaches are designed for the supervised scenario with sufficient observations of the response $\by$ to enable high-quality learning and inference of the graphical models. In the SSL scenario with only a small amount of labeled samples and large unlabeled data with auxiliary features, no methods have been developed yet for efficient statistical analysis of Ising models. Our proposed method actually forms an SSL complement to the above-reviewed literature and can be potentially incorporated with them.

There has notable methodological progression made on semi-supervised statistical learning and inference in the past years. \cite{kawakita2013semi} and \cite{kawakita2014safe} developed density ratio approaches first fitting a importance weighting model between the labeled and unlabeled samples, whose underlying truth is known to be trivially $1$, then using it to weight the regression on labeled sample. The efficiency gain of these methods is a paradox previously studied in the semiparametric literature \citep[e.g.]{robins1992estimating}. \cite{chakrabortty2018efficient} proposed an imputation-based SSL approach that incorporates the unlabeled sample for linear regression after imputing its unobserved outcome with the auxiliary features. Their method is adaptive and safe in the sense that it will not result in larger variance than the supervised estimator. \cite{azriel2022semi} achieved a similar adaptive property with a more simple and efficient SSL method. \cite{gronsbell2022efficient} extended the idea of imputation to handle the logistic model and its validation. We also note other recent advances in this track. For example, \cite{chakrabortty2022semi} addressed SSL quantile regression and \cite{tony2020semisupervised} and \cite{zhang2022high} considered the SSL estimation of the outcome's mean and variance with high-dimensional auxiliary features. However, no existing work in this track addressed or could be directly extended for the SSL of graphical models, with the technical complication and challenges discussed later.

Typically, clinical outcomes like laboratory measures and diagnostic codes are viewed as error-prone yet informative surrogates for the true disease outcomes of our interest. These surrogates can play a central role in EHR driven biomedical studies, especially when obtaining the true or primary outcomes requires chart reviewing or long-term follow-up. As a special case of SSL, surrogate-assisted SSL (SAS) leverages the surrogate outcomes observed on all samples to assist the learning with true outcomes. In specific, \cite{huang2018pie}, \cite{hong2019semi}, and \cite{zhang2020maximum} effectively leveraged the EHR surrogates assuming they are bimodal and independent with the baseline risk factors given the underlying true conditions. \cite{athey2020combining} proposed a new method to combine the observational and experimental data using their share surrogate outcomes, to improve the efficiency of casual inference. \cite{hou2021efficient} considered the SAS of average treatment effect where both the treatment and outcome are approximated using EHR surrogates on the unlabeled sample. \cite{hou2021efficient} extended the imputation-based SSL approach to utilize the EHR surrogates for efficient high-dimensional sparse regression. \cite{zhang2022prior} realized adaptive SAS sparse regression via information-guided regularization. Their method is more robust to EHR surrogates with poor ability to characterize the true outcome. Built upon a general SSL framework for Ising models, our work will highlight its implementation under the SAS setting, with a primary application in EHR studies, and inspired by some existing SAS literature \citep[e.g.]{hong2019semi}.

\subsection{Our contribution}

In this paper, we develop a novel SCore-Induced Semi-Supervised (SCISS) learning approach for Ising models. The core idea of SCISS is to jointly model the outcomes $\by$ against the auxiliary EHR features $\bz$ with the labeled data $\mathscr{L}$, then use it to project the score function of the SL  estimator onto the space of $\bz$, and incorporate both $\mathscr{L}$ and $\mathscr{U}$ to augment SL, which can effectively reduce the variance of the projected part in SL determined by $\bz$ without introducing bias. The small labeled data may not afford fitting a joint model for $\by\sim\bz$ with enough high complexity, which could result in potential misspecification on $\by\sim\bz$ and efficiency loss. To maintain the robustness of our estimator to this issue, we leverage the nature of EHR data to propose two modeling strategies that avoid excessive model complexity while maintaining good performance. The first one is an augmented Ising model with coefficients varying with the EHR features and the second one is built upon the post-hoc generative assumption of EHR surrogates. In addition, we propose two strategies, intrinsic efficient adjustment and optimal allocation, both ensuring our SSL estimator to have no larger variance than the SL estimator.


To our best knowledge, our work is the first one to realize robust and efficient SSL of Ising model in the context of statistical learning. Existing imputation and other approaches for SSL like \cite{chakrabortty2018efficient} and \cite{azriel2022semi} cannot be adopted for graphical models as the score function of the SL estimator is non-linear and non-separable in $\by$. We address this by extending the outcome-imputation strategy in \cite{chakrabortty2018efficient} to more general score-induced projection strategy. Also, our proposed methods in Sections \ref{sec:Aug} and \ref{sec:PoS} serve as novel tools for utilizing EHR surrogates to jointly characterize multiple disease outcomes in SSL. These methods, as well as the intrinsic efficient estimation procedures in Section \ref{sec:misspecification}, are elaborately designed in order to avoid excessive model complexity and achieve robust and efficient SSL for Ising models.

\section{Method}\label{sec:method}

\subsection{Preliminary: supervised estimation}\label{sec:method:sup}
 
In Section \ref{sec:method}, we will propose the procedures to obtain our SCore-Induced Semi-Supervised (SCISS) estimator for the Ising model parameters $\btheta$. We shall start with the standard SL estimator $\widehat\btheta_{\rm SL}$ obtained only using the labeled sample. Motivated by the conditional logistic model (\ref{equ:Ising_meaning}) derived from (\ref{equ:Ising model}), we solve for each $\widecheck\btheta_{j,{\rm SL}}$ with labeled $\mathscr{L}$ through
\begin{equation}\label{equ:SL}
    \mathbf{U}_{j,n}(\btheta_j)=\frac{1}{n}\sum_{i=1}^n \by^{i}_{\setminus j,\bw}\{y_j^i-g(\btheta_j\trans\by^{i}_{\setminus j,\bw})\}=\bzero,{~}j=1,\ldots,q,
\end{equation}
where $g(a)=e^a/(1+e^a)$. Estimating equations (\ref{equ:SL}) is equivalent with the logistic regression of $y_j$ against the remaining vector of variables $\by^{i}_{\setminus j,\bw}$ based on the maximum likelihood estimation (MLE). Due to the underlying mirroring property of $\theta_{jk}$ and $\theta_{kj}$ in $\btheta$, we obtain the final SL estimator $\widehat\btheta_{\rm SL}$ by symmetrizing on $\widecheck\btheta_j$ and $\widecheck\btheta_k$ as $\widehat\theta_{jk,{\rm SL}}=\frac{1}{2}(\widecheck\theta_{jk,{\rm SL}}+\widecheck\theta_{kj,{\rm SL}})$ for all $j\neq k$ and $\widehat\btheta_{jj,{\rm SL}}=\widecheck\btheta_{jj,{\rm SL}}$ for $j=1,\ldots,q$.

Denote the population (true) model parameters as $\overline\btheta$. Following the well-established and standard theory in M-estimation \citep{van2000asymptotic}, we present Lemma \ref{lemma:SL} below, which establishes the asymptotic unbiasedness and normality of $\widehat\btheta_{\rm SL}$ and gives the form of the score function of $\widehat\btheta_{\rm SL}$ determining its asymptotic variance. For the simplicity of notation, we use $\by_{\bw}$ to denote the composition of $\{\by,\bw\}$ when considering the function on $\mathscr{L}$ without the auxiliary $\bx$. 


\begin{lemma}
Under regularity Conditions \ref{cond:1} and \ref{cond:2} {\rm (A)} introduced in Appendix \ref{sec:proof_SL}, we have
\[
n^{\frac{1}{2}}(\widehat{\theta}_{jk,{\rm SL}}-\overline\theta_{jk})=\frac{1}{2}n^{-\frac{1}{2}}\sum_{i=1}^n\{ s_{kj}(\by_{\bw}^i;\overline\btheta_k)+ s_{jk}(\by_{\bw}^i;\overline\btheta_j)\}+o_p(1)
\]
weakly converges to $N(0,\Omega_{jk,{\rm SL}})$, where $s_{jk}(\by_{\bw};\overline\btheta_j)$ is the $k$-th element of the score function of equations (\ref{equ:SL}) $\S_j(\by_{\bw};\btheta_j)=\Sigma_{\btheta_j}^{-1}\by_{\setminus j,\bw}\{y_j-g(\btheta_j\trans\by_{\setminus j,\bw})\}$ for $j=1,2,\ldots,q$, with the hessian matrix $\Sigma_{\btheta_j}={\rm E}[\by_{\setminus j,\bw}\by_{\setminus j,\bw}\trans\dot g(\btheta_j\trans\by_{\setminus j,\bw})]$ and $\Omega_{jk,{\rm SL}}=\frac{1}{4}{\rm E}[\{s_{jk}(\by_{\bw};\overline\btheta_j)+s_{kj}(\by_{\bw};\overline\btheta_k)\}^2]$.
\label{lemma:SL}      
\end{lemma}
Asymptotic expansion in Lemma \ref{lemma:SL} implies that the variance of $\widehat{\theta}_{jk,{\rm SL}}$ is driven by those of the zero-mean score functions $\S_j(\by_{\bw};\overline\btheta_j)$ and $\S_k(\by_{\bw};\overline\btheta_k)$. This inspires us to reduce the variance of each $\widehat{\theta}_{jk,{\rm SL}}$ through reducing the variance of $n^{-\frac{1}{2}}\sum_{i=1}^n\{ s_{kj}(\by_{\bw}^i;\overline\btheta_k)+ s_{jk}(\by_{\bw}^i;\overline\btheta_j)\}$ leveraging both $\mathscr{L}$ and $\mathscr{U}$. In Section \ref{sec:framework}, we will present the corresponding ideas and details.

\subsection{SCISS framework}\label{sec:framework}


Based on the auxiliary EHR features $\bx$ that are informative to $\by$ and $\S_j(\by_{\bw};\btheta_j)$, our idea is to model $\by$ against $\bz$ with the labeled $\mathscr{L}$, use it to project $\S_j(\by_{\bw};\overline\btheta_j)$ onto the space of $\bz$, and incorporate the projection on $\mathscr{L}$ and $\mathscr{U}$ to augment $\widehat\btheta_{\rm SL}$, which can effectively reduce the variance of the projected part determined by $\bz$ without introducing bias.

In specific, we introduce $P(\overline\by|\bz;\bfeta)=P(\by=\overline\by|\bz)$ for any given $\overline\by\in \{0,1\}^q$ as the conditional probability model for $\by=\overline\by$ given $\bz$ with parameters $\bfeta$. The specific forms and choices on $P(\overline\by|\bz;\bfeta)$ will be introduced with details in Section \ref{sec:impute}. Also, as will be discussed later, our model $P(\overline\by|\bz;\bfeta)$ could be misspecified or of excessive estimation error rate thanks to our bias correction construction. Then we denote by $\m_j(\bz)={\rm E}[\S_j(\by_{\bw};\overline\btheta_j)|\bz]$, and estimate $\m_j(\bz)$ through
\begin{equation}
\widehat\m_j(\bz;\widehat\bfeta,\widehat\btheta_{j,{\rm SL}})=\sum_{\overline\by\in \{0,1\}^q} P(\overline\by|\bz;\widehat\bfeta)\cdot\widehat\S_j(\overline\by_{\bw};\widehat\btheta_{j,{\rm SL}}),
\label{equ:imputation}
\end{equation} 
where $\widehat\bfeta$ is an estimator of $\bfeta$ obtained with the labeled data, and $\widehat\S_j(\by_{\bw};\widehat\btheta_{j,{\rm SL}})={\widehat\Sigma_{\btheta_j}}^{-1}\by_{\setminus j,\bw}\{y_j-g(\btheta_j\trans\by_{\setminus j,\bw})\}$ is an empirical estimate of the score function with $\widehat\Sigma_{\btheta_j}=\frac{1}{n}\sum_{i=1}^n\by_{\setminus j,\bw}^i(\by_{\setminus j,\bw}^i)\trans\dot g(\btheta_j\trans\by^i_{\setminus j,\bw})$. The fitted $\widehat\m_j(\bz;\widehat\bfeta,\widehat\btheta_{j,{\rm SL}})$ can be viewed as an empirical projection of the influence function of $\widecheck\btheta_{j,{\rm SL}}$ onto the features $\bz$ observed on the large sample, neglecting the errors in $\widehat\bfeta$ and $\widehat\btheta_{j,{\rm SL}}$.




Motivated by Lemma \ref{lemma:SL}, to reduce the variance of $\widecheck\btheta_{j,{\rm SL}}$, it is tentative to adjust the SL estimator as
$\widecheck\btheta_{j,{\rm SL}}-n^{-1}\sum_{i=1}^n\widehat\m_j(\bz^i;\widehat\bfeta,\widehat\btheta_{j,{\rm SL}})$, with its expansion becoming approximately
\[
n^{-1}\sum_{i=1}^n \S_j(\by^i_{\bw};\overline\btheta_j)-\widehat\m_j(\bz^i;\widehat\bfeta,\widehat\btheta_{j,{\rm SL}}),
\]
which could potentially achieve variance reduction since $\m_j(\cdot)$ is the projection (conditional mean) of $\S_j(\cdot)$ on $\bz$. However, this construction itself is problematic since the bias could arise from the misspecification of $P(\overline\by|\bz;\bfeta)$ or the estimation errors in $\widehat\bfeta$ and $\widehat\btheta_{j,{\rm SL}}$, which makes the added term $\widehat\m_j(\bz^i;\widehat\bfeta,\widehat\btheta_{j,{\rm SL}})$ not zero-mean. In response to this issue, we incorporate the unlabeled sample $\mathscr{U}$ to construct 
\begin{equation}
        \widecheck\btheta_{j,{\rm SCISS}}=\widecheck\btheta_{j,{\rm SL}}-n^{-1}\sum_{i=1}^n \widehat\m_j(\bz^i;\widehat\bfeta,\widehat\btheta_{j,{\rm SL}})+N^{-1}\sum_{i=n+1}^{n+N}\widehat\m_j(\bz^i;\widehat\bfeta,\widehat\btheta_{j,{\rm SL}}),
        \label{equ:construction}
\end{equation}
with the mean of $\widehat\m_j(\bz^i;\widehat\bfeta,\widehat\btheta_{j,{\rm SL}})$ evaluated on  $\mathscr{U}$ to make up for the above-mentioned bias. Note that $\widehat\m_j(\bz^i;\widehat\bfeta,\widehat\btheta_{j,{\rm SL}})$ can be calculated on $\mathscr{U}$ because it is free of the true outcomes $\by$. Again, due to the mirroring property that $\theta_{jk}=\theta_{kj}$, we obtain the final SCISS estimator $\widehat\btheta_{\rm SCISS}$ as $\widehat\theta_{jk,{\rm SCISS}}=\frac{1}{2}(\widecheck\theta_{jk,{\rm SCISS}}+\widecheck\theta_{kj,{\rm SCISS}})$ for $j\neq k$, and $\widehat\theta_{jj,{\rm SCISS}}=\widecheck\theta_{jj,{\rm SCISS}}$ for each $j$.

Given the asymptotic unbiasedness of the SL estimator established in Lemma \ref{lemma:SL}, our construction in (\ref{equ:construction}) ensures $\widehat\theta_{jk,{\rm SCISS}}$ to be also unbiased because both the non-zero mean of $\widehat\m_j(\bz^i;\widehat\bfeta,\widehat\btheta_{j,{\rm SL}})$ under misspecified $P(\overline\by|\bz;\bfeta)$ and the estimation errors of the nuisance $\widehat\bfeta$ and $\widehat\btheta_{j,{\rm SL}}$ can be canceled out between the two empirical mean terms on $\mathscr{U}$ and $\mathscr{L}$ in (\ref{equ:construction}). In this sense, the difference of these two terms can be viewed as a control variate for $\widecheck\btheta_{j,{\rm SL}}$. Moreover, when our model for $\by|\bz$ is correct, we can show that $\widehat\theta_{jk,{\rm SCISS}}-\overline\theta_{jk}$ is asymptotically equivalent to
\[
\frac{1}{2n}\sum_{i=1}^n\left\{ s_{kj}(\by_{\bw}^i;\overline\btheta_k)+ s_{jk}(\by_{\bw}^i;\overline\btheta_j)-{\rm E}\left[s_{kj}(\by_{\bw};\overline\btheta_k)+ s_{jk}(\by_{\bw};\overline\btheta_j)|\bz=\bz^i\right]\right\},
\]
since the variance term from $\mathscr{U}$ is asymptotically negligible as $N\gg n$. Thus, $\widehat\theta_{jk,{\rm SCISS}}$ can attain strictly smaller variance than $\widehat\theta_{jk,{\rm SL}}$. See Theorem \ref{thm:SSL} and its justification in Appendix \ref{sec:proof_SSL} for details. This result is also closely connected to the theory of efficient influence function in semiparametric literature \citep[e.g.]{hines2022demystifying}. 

However, there are still two technical problems for our framework remained to be solved. One is how to properly specify the conditional model $P(\overline\by|\bz;\bfeta)=P(\by=\overline\by|\bz)$ when the small set $\mathscr{L}$ could not afford us to learn $\by\sim\bz$ nonparametrically or through complex machine learning methods. In this case, it is preferred to leverage the generative relationship between the phenotypes and EHR features for parametric modelling without excessive complexity while maintaining effectiveness, the details of which will be provided in Section \ref{sec:impute}. The second is how to achieve the optimal improvement on SCISS over the SL estimator when our model for $\by|\bz$ is potentially misspecified, at least the safety in the sense that SCISS is always not worse than SL. Our developments for this purpose will be introduced in Section \ref{sec:misspecification}. Incorporated with these techniques, our SCISS framework can provide a robust and efficient estimator for the Ising model with the help of large $\mathscr{U}$, as will be justified in our asymptotic and numerical studies.

\subsection{Construct the conditional model of $\by$}\label{sec:impute}

\subsubsection{Augmented Ising model}\label{sec:Aug}
 
In this section, we introduce our first constructing strategy of $P(\overline\by|\bz;\bfeta)$, an augmented Ising model with its coefficients varying with the EHR features. In various real-life scenarios, the network structure of $\by$, i.e., the conditional association between each pair $y_j$ and $y_k$ may vary with some extraneous factors available to us, just like the auxiliary features $\bx$ in our setting. For example, a diagnostic or laboratory measure indicating high low-density lipoprotein (LDL) cholesterol may encode a larger chance of co-occurrence of heart disease and stroke, both of which are consequences of high LDL. Motivated by this, we propose a conditional probability model augmenting the original Ising model of $\by$ in (\ref{equ:Ising model}) with $\bx$ as:  
\begin{equation}
    P(\by|\bz;\bfeta)=\frac{1}{Z(\bz;\bfeta)}\exp\bigg(\sum_{j=1}^q{\bfeta}_{jj}(\bz)y_j+\sum_{(j,k):1\le j < k  \le q}{\bfeta}_{jk}(\bx)y_j y_k\bigg),
    \label{equ:Augmented-Ising}
\end{equation}
 where $\bfeta(\bz)=\{\bfeta_{11}(\bz),\bfeta_{12}(\bx),\ldots,\bfeta_{(q-1)q}(\bx),\bfeta_{qq}(\bz)\}$ and $Z(\bz;\bfeta)$ is the partition function. A natural way to model $\bfeta(\bz)$ is to parametrize it as a linear function of $\bz$. Specifically, we let $\bfeta_{jk}(\bx)=\bx\trans\bfeta_{jk}$, where $\bfeta_{jk}=(\eta_{jk1},\ldots,\eta_{jkp})\trans$ for $j\neq k$ and ${\bfeta}_{jj}(\bz)=(\bx\trans,\bw\trans)\bfeta_{jj}$, $\bfeta_{jj}=(\eta_{jk1},\ldots,\eta_{jkp},\eta_{jk(p+1)},\ldots,\eta_{jk(p+d+1)})\trans$ for $j=1,\ldots,q$. Similar to (\ref{equ:Ising_meaning}), the conditional model for each $y_j$ under (\ref{equ:Augmented-Ising}) is a logistic regression model with its coefficients varying with $\bx$:
  \[
    \log\bigg(\frac{P(y_j=1|\by_{\setminus j,\bw},\bz)}{1-P(y_j=1|\by_{\setminus j,\bw},\bz)}\bigg)=\bfeta_{jj}(\bz)+\sum_{k:k\neq j}\bfeta_{jk}(\bx)y_k.
 \]
Thus, we estimate the augmented Ising model's parameters in $\bfeta$ by solving each
\begin{equation*}
\mathbf{Q}^{\rm Aug}_{j,n}(\bfeta_{j})
=\frac{1}{n}\sum_{i=1}^n \boldsymbol{\psi}^{i}_{\setminus j,\bw}\{y_j^i-g(\bfeta_{j}\trans\boldsymbol{\psi}^{i}_{\setminus j,\bw})\}+\lambda_n\bfeta_{j}=\bzero,
\end{equation*}
and denote the solution as $\widehat\bfeta_{j}=(\widehat\bfeta_{1}\trans,\ldots,\widehat\bfeta_{q}\trans)\trans$, where $\lambda_n$ is a small tuning parameter to ensure stable training (e.g., taking as the $o(n^{-\frac{1}{2}})$ rate), and
 \[\boldsymbol{\psi}_{\setminus j,\bw}=(\bx\trans y_1,\ldots,\bx\trans y_{j-1},(\bx\trans,\bw\trans),\bx\trans y_{j+1},\ldots,\bx\trans y_q)\trans.\]


\subsubsection{Post-hoc surrogate model}\label{sec:PoS}

Next, we introduce our second strategy to extract $P(\overline\by|\bz;\bfeta)$, which is motivated by and designed for post-hoc surrogates in EHR. We consider the scenario with the number of auxiliary features $p=q$ and for each outcome $y_j$, there is a main surrogate variable $x_j$ informative of $y_j$ and satisfying that $x_j\perp \by_{\setminus j}| \{y_j,\bw\}$ where $\by_{\setminus j}=(y_1,\ldots,y_{j-1},y_{j+1},\ldots,y_q)\trans$. Taking type II diabetes (T2D) and stroke as an example, these two chronic diseases may co-occur in some period and result in hospital visits and increased counts of their diagnostic (ICD) codes in $\bx$. Meanwhile, given the status of T2D, its main surrogate, i.e. the ICD count tends to be not dependent of other diseases like stroke, corresponding to our conditional independence assumption $x_j\perp \by_{\setminus j}| \{y_j,\bw\}$. Using this assumption, we can formulate the model for $\by\sim\bz$ through Bayes' theorem:
\begin{equation}
    P(\by|\bz;\bfeta)=\frac{P(\bx,\by|\bw;\bfeta)}{P(\bx|\bw;\bfeta)}=\frac{P(\bx|\by,\bw;\bxi)P(\by|\bw;\btheta)}{\sum_{\overline\by\in\{0,1\}^q}P(\bx,\overline\by|\bw;\bfeta)}=\frac{P(\by|\bw;\btheta)\prod_{j=1}^q P(x_j|y_j,\bw;\bxi_j)}{\sum_{\overline\by\in\{0,1\}^q}P(\bx,\overline\by|\bw;\bfeta)},
    \label{equ:PoS_Bayesian}
\end{equation}
where $\bfeta=\{\btheta,\bxi\}$ with $\btheta$ and $\bxi=\{\bxi_j:j=1,2,\dots,q\}$ used for the two components on the right hand side of $P(\bx,\by|\bw;\bfeta)=P(\bx|\by,\bw;\bxi)P(\by|\bw;\btheta)$. By $x_j\perp \by_{\setminus j}| \{y_j,\bw\}$, we can further decompose the term $P(\bx|\by,\bw;\bxi)$ in (\ref{equ:PoS_Bayesian}) as $\prod_{j=1}^q P(x_j|y_j,\bw;\bxi_j)$. 

Based on (\ref{equ:PoS_Bayesian}), we are able to characterize $P(\by|\bz;\bfeta)$ by taking $\btheta$ as $\widehat\btheta_{\rm SL}$ obtained in Section \ref{sec:method:sup} and $P(\bx|\by,\bw;\bxi)$ as $P(\bx|\by,\bw;\bxi)=\Pi_{j=1}^q P(x_j|y_j,\bw;\widehat\bxi_j)$ where each $\widehat\bxi_j$ is derived using the labeled data $\mathscr{L}$. Here, we shall list examples of $P(x_j|y_j,\bw;\bxi_j)$ for different types of frequently encountered EHR surrogates $\bx$. (I) Gaussian linear models for continuous surrogates: $x_j\mid \{y_j,\bw\}\sim N\{(\bw\trans,y_j)\bxi^{\dagger}_j, \sigma^2_j\}$ and $\bxi_j=\{\bxi^{\dagger}_j,\sigma^2_j\}$; (II) Logistic models for binary $x_j$: $P(x_j=1|y_j,\bw;\bxi_j)=g\{(\bw\trans,y_j)\bxi_j\}$; (III) Poisson models for counting surrogate: $x_j\sim {\rm Poi}({\rm exp}\{(\bw\trans,y_j)\bxi_j)\})$.


At last, we establish the asymptotic properties of $\widehat\btheta_{\rm SCISS}$ in Theorem \ref{thm:SSL} and the proof is provided in Appendix \ref{sec:proof_SSL}. Based on this, we justify in Corollary \ref{coro:1} that when the model for $\by|\bz$ is correctly specified and not null, $\widehat\btheta_{\rm SCISS}$ is strictly more efficient than $\widehat\btheta_{\rm SL}$, corresponding to our heuristic discussion in Section \ref{sec:framework}. Also, we can quantify the uncertainty of $\widehat\btheta_{\rm SCISS}$ through the empirical version of its standard error given in Theorem \ref{thm:SSL}; see equation (\ref{equ:var}) in the next section. Note that Theorem \ref{thm:SSL} is applicable to both Augmented Ising model and Post-hoc surrogate model, both of which are justified in Appendix \ref{sec:proof_SSL} as well.
\begin{theorem}
\label{thm:SSL}
Under Condition $\ref{cond:1}-\ref{cond:2}$, $\widehat\btheta_{\rm SCISS}\xrightarrow{p}\overline\btheta$, and
\begin{equation*}
    \begin{aligned}
    n^{\frac{1}{2}}(\widehat{\theta}_{jk,{\rm SCISS}}-\overline\theta_{jk})=&\frac{1}{2}n^{-\frac{1}{2}}\sum_{i=1}^{n}\{ s_{jk}(\by_{\bw}^i;\overline\btheta_j)+ s_{kj}(\by_{\bw}^i;\overline\btheta_k)- m_{jk}(\bz^i;\overline\bfeta,\overline\btheta_j)- m_{kj}(\bz^i;\overline\bfeta,\overline\btheta_k)\}\\
    &+\frac{1}{2}n^{\frac{1}{2}}N^{-1}\sum_{i=n+1}^{n+N}\{ m_{jk}(\bz^i;\overline\bfeta,\overline\btheta_j)+ m_{kj}(\bz^i;\overline\bfeta,\overline\btheta_k)\}+o_p(1)
    \end{aligned}
\end{equation*}
weakly converges to $N(0,\Omega_{jk,{\rm SCISS}})$, where 
\[
\Omega_{jk,{\rm SCISS}}=\frac{1}{4}{\rm Var}\left\{s_{jk}(\by_{\bw};\overline\btheta_j)+s_{kj}(\by_{\bw};\overline\btheta_k)-m_{jk}(\bz;\overline\bfeta,\overline\btheta_j)-m_{kj}(\bz;\overline\bfeta,\overline\btheta_k)\right\},
\]
and $m_{jk}(\bz;\overline\bfeta,\overline\btheta_j)$ is the $k$-th element of $\m_j(\bz;\overline\bfeta,\overline\btheta_j)$.
\end{theorem}

\begin{corollary}
When the conditional model for $\by$ given $\bz$ is correctly specified, i.e. $m_{jk}(\bz;\overline\bfeta,\overline\btheta_j)$ is the true conditional mean of $s_{jk}(\by_{\bw};\overline\btheta_j)$ on $\bz$, we have $\Omega_{jk,{\rm SL}}\geq \Omega_{jk,{\rm SCISS}}$ for $j,k=1,\ldots,q$, where the equality between $\Omega_{jk,{\rm SL}}$ and $\Omega_{jk,{\rm SCISS}}$ only holds when the variance of $m_{jk}(\bz;\overline\bfeta,\overline\btheta_j)$ is zero.
\label{coro:1}
\end{corollary}

\section{Overcome model misspecification}\label{sec:misspecification}

Corollary \ref{coro:1} establishes that SCISS is more efficient than the SL estimator when our model for $\by|\bz$ in Section \ref{sec:Aug} or \ref{sec:PoS} is correctly specified and $m_{jk}(\bz;\overline\bfeta,\overline\btheta_j)$ is the conditional mean of $s_{jk}(\by_{\bw};\overline\btheta_j)$ on $\bz$. Nevertheless, under misspecified $\by|\bz$, $\widehat\theta_{jk, {\rm SCISS}}$ defined in Section \ref{sec:method} is not guaranteed to have an essentially smaller variance than SL, and could be even of a larger one than the latter occasionally. In this section, we will propose two approaches to further enhance the performance of SCISS under potentially misspecified models on $\by|\bz$, and ensure it to have no larger variance than the SL estimator in whatever cases.



 
\subsection{Intrinsic efficient SS estimation}\label{sec:intrinsic}

Motivated by Theorem \ref{thm:SSL}, the asymptotic variance of our SCISS estimator on each $\theta_{jk}$ is a function of the nuisance conditional model parameter $\bfeta$, with its empirical version formed as 
\begin{equation}
\begin{split}
    \widehat\sigma^2_{jk}(\bfeta)=&\frac{1}{n}\sum_{i=1}^n\left\{\widehat s_{jk}(\by_{\bw}^i;\widehat\btheta_{j,{\rm SL}})+\widehat s_{kj}(\by_{\bw}^i;\widehat\btheta_{k,{\rm SL}})-\widehat m_{jk}(\bz^i;\bfeta,\widehat\btheta_{j,{\rm SL}})-\widehat m_{kj}(\bz^i;\bfeta,\widehat\btheta_{k,{\rm SL}})\right\}^2\\
    &+\left\{\frac{1}{n}\sum_{i=1}^n\widehat s_{jk}(\by_{\bw}^i;\widehat\btheta_{j,{\rm SL}})+\widehat s_{kj}(\by_{\bw}^i;\widehat\btheta_{k,{\rm SL}})-\widehat m_{jk}(\bz^i;\bfeta,\widehat\btheta_{j,{\rm SL}})-\widehat m_{kj}(\bz^i;\bfeta,\widehat\btheta_{k,{\rm SL}})\right\}^2
\end{split}
    \label{equ:var}
\end{equation}
no matter the model for $\by\sim\bz$ is correct or not. Interestingly, $\widehat\sigma^2_{jk}(\bfeta)$ will become the empirical variance of the SL estimator when our conditional model $P(\by|\bz;\bfeta)$ parameterized by $\bfeta$ degenerates to $P(\by|\bw;\btheta)$, the distribution of $\by|\bw$. This happens when all coefficients for the auxiliary $\bx$ equal to zero in the augmented Ising model introduced in Section \ref{sec:Aug}, and when $\bxi=\bzero$ for the post-hoc surrogate model in Section \ref{sec:PoS}. In both cases, there exists some $\bfeta_{{\rm null}}$ such that $\widehat\sigma^2_{jk}(\bfeta_{{\rm null}})$ is consistent to $\Omega_{jk,{\rm SL}}$, the asymptotic variance of $n^{\frac{1}{2}}(\widehat{\theta}_{jk,{\rm SL}}-\overline\theta_{jk})$ given by Lemma \ref{lemma:SL}.

Inspired by our above discussion and existing literature in intrinsic efficient (INTR) estimation  like \cite{rotnitzky2012improved} and \cite{gronsbell2022efficient}, our main idea is to use the $\bfeta$ directly minimizing $\widehat\sigma^2_{jk}(\bfeta)$ to construct SCISS, instead of using the MLE of the potentially misspecified models in Sections \ref{sec:Aug} and \ref{sec:PoS}. In specific, we first try to extract
\begin{equation}
    \widehat\bfeta_{jk,\rm INTR}=\mathop{\rm argmin}_{\bfeta}\widehat\sigma_{jk}^2(\bfeta), 
\label{equ:intrinsic:eta}    
\end{equation}
then plug $\widehat\bfeta=\widehat\bfeta_{jk,\rm INTR}$ in (\ref{equ:construction}) to obtain $\widecheck\btheta_{j,{\rm INTRI}}$, and finally take $\widehat\theta_{jk,{\rm INTR}}=\frac{1}{2}(\widecheck\theta_{jk,{\rm INTR}}+\widecheck\theta_{kj,{\rm INTR}})$ again from symmetricity. In this way, we ensure the variance of the SCISS estimator to be not larger than SL as a special case with $\bfeta_{{\rm null}}$ in our framework. For asymptotic analysis of $\widehat\btheta_{jk, {\rm INTR}}$, we define $\overline\bfeta_{jk,{\rm INTR}}=\mathop{\rm argmin}\sigma_{jk}^2(\bfeta)$, where $\sigma_{jk}^2(\bfeta)={\rm Var}\{s_{jk}(\by_{\bw};\overline\btheta_j)+s_{kj}(\by_{\bw};\overline\btheta_k)-m_{jk}(\bz;\bfeta,\overline\btheta_j)-m_{kj}(\bz;\bfeta,\overline\btheta_k)\}$, and present Theorem \ref{thm:intri} that is justified in Appendix \ref{sec:proof_intri}.
\begin{theorem}
Under Condition \ref{cond:1}-\ref{cond:intri} introduced in Appendix \ref{sec:proof_intri}, $\widehat\btheta_{\rm INTR}\xrightarrow{p}\overline\btheta$, and
\begin{equation*}
\begin{aligned}
n^{\frac{1}{2}}(\widehat\theta_{jk,{\rm INTR}}-\overline\theta_{jk})=&\frac{1}{2}n^{-\frac{1}{2}}\sum_{i=1}^{n}\{s_{jk}(\by_{\bw}^i;\overline\btheta_j)+ s_{kj}(\by_{\bw}^i;\overline\btheta_k)- m_{jk}(\bz^i;\overline\bfeta_{jk,{\rm INTR}},\overline\btheta_j)- m_{kj}(\bz^i;\overline\bfeta_{kj,{\rm INTR}},\overline\btheta_k)\}\\
&+\frac{1}{2}n^{\frac{1}{2}}N^{-1}\sum_{i=n+1}^{n+N}\{ m_{jk}(\bz^i;\overline\bfeta_{jk,{\rm INTR}},\overline\btheta_j)+ m_{kj}(\bz;\overline\bfeta_{kj,{\rm INTR}},\overline\btheta_k)\}+o_p(1).
\end{aligned}
\end{equation*}
(i) When the model $P(\by|\bz;\bfeta)$ is correctly specified for $\by|\bz$, $\widehat\theta_{jk,{\rm INTR}}$ is asymptotically equivalent to $\widehat\theta_{jk,{\rm SCISS}}$. (ii) The asymptotic variance of $n^{\frac{1}{2}}(\widehat\theta_{jk,{\rm INTR}}-\overline\theta_{jk})$ is minimized among the estimators constructed under the model $P(\by|\bz;\bfeta)$. Consequently, it is always less than or equal to the asymptotic variance of both $n^{\frac{1}{2}}(\widehat\theta_{jk,{\rm SL}}-\overline\theta_{jk})$ and $n^{\frac{1}{2}}(\widehat\theta_{jk,{\rm SCISS}}-\overline\theta_{jk})$ constructed using the methods in Section \ref{sec:Aug} or \ref{sec:PoS}.
\label{thm:intri}
\end{theorem}

Theorem \ref{thm:intri} establishes the optimality of $\widehat\theta_{jk,{\rm INTR}}$ under certain specification of the model for $\by|\bz$ that is potentially wrong. This also implies an improvement of efficiency by using $\widehat\theta_{jk,{\rm INTR}}$ over the SL and original SCISS estimators. However, it is important to note that solving for $\widehat\bfeta_{jk,{\rm INTR}}$ in (\ref{equ:intrinsic:eta}) is a non-convex problem and one may encounter local minima issues impeding $\widehat\theta_{jk,{\rm INTR}}$ to achieve its optimality in practice. In response to this, we introduce a safe optimization strategy on (\ref{equ:intrinsic:eta}) that initializes from the MLE of $\bfeta$ introduced in Section \ref{sec:Aug} or \ref{sec:PoS}, then iterates on updating $\bfeta$ via the Newton-Raphson method, and stops the iteration whenever converging or the objective function $\widehat\sigma_{jk}^2(\bfeta)$ does not decrease anymore. We provide the implementation details in Appendix \ref{sec:procedure_intri} and a simulation study of $\widehat\theta_{jk,{\rm INTR}}$ in Appendix \ref{sec:intrinsic simulation}. 


\subsection{Ensemble estimation}\label{sec:ensemble}

In addition to intrinsic efficient estimation, we shall introduce an ensemble estimation (ES) strategy based on the optimal allocation, which is easier to implement than INTR but practically performs well in power enhancement under potentially misspecified models for $\by|\bz$. The key idea is that an appropriately weighted convex combination of several asymptotically unbiased and (jointly) normal estimators is still unbiased and can have no larger variance than each of them. In specific, we consider the component-wise ensemble of $\widehat\theta_{jk,{\rm SCISS}}$ and $\widehat\theta_{jk,{\rm SL}}$ formulated as
\[
    \widehat\theta_{jk,{\rm ES}}=\widehat\alpha_{jk}\widehat\theta_{jk,{\rm SCISS}}+(1-\widehat\alpha_{jk})\widehat\theta_{jk,{\rm SL}},
\]
where $\widehat\alpha_{jk}\in[0,1]$ is an allocation parameter determined by us from the data, to minimize the variance of $\widehat\theta_{jk,{\rm ES}}$. For this purpose, we take $\widehat\alpha_{jk}$ as the first component of the vector $\mathbf{1}\trans{\widehat\Lambda_{jk}}^{-1}/(\mathbf{1}\trans{\widehat\Lambda_{jk}}^{-1}\mathbf{1})$ and $\widehat\Lambda_{jk}$ is the moment estimator of the covariance of the score functions of $\widehat\theta_{jk,{\rm SL}}$ and $\widehat\theta_{jk,{\rm SCISS}}$ respectively given in Lemma \ref{lemma:SL} and Theorem \ref{thm:SSL}, which is actually proportional to the asymptotic covariance of  $\widehat\theta_{jk,{\rm SL}}$ and $\widehat\theta_{jk,{\rm SCISS}}$. In this way, $\widehat\theta_{jk,{\rm ES}}$ will be guaranteed a variance smaller than or equal to both $\widehat\theta_{jk,{\rm SCISS}}$ and $\widehat\theta_{jk,{\rm SL}}$ since they correspond respectively to $\widehat\alpha_{jk}=1$ and $\widehat\alpha_{jk}=0$ in our ES construction. Furthermore, we can also construct the ES estimation as 
\[
\widehat\theta_{jk,{\rm ES}}=\widehat\alpha_{jk, {\rm SL}}\widehat\theta_{jk, {\rm SL}}+\widehat\alpha_{jk, {\rm Aug}}\widehat\theta_{jk, {\rm SCISS}}^{\rm Aug}+\widehat\alpha_{jk, {\rm PoS}}\widehat\theta_{jk, {\rm SCISS}}^{\rm PoS},
\]
where $(\widehat\alpha_{jk, {\rm SL}},\widehat\alpha_{jk, {\rm Aug}},\widehat\alpha_{jk, {\rm PoS}})$ are again the optimal allocation obtained through variance minimization subject to $\widehat\alpha_{\rm SL}+\widehat\alpha_{\rm Aug}+\widehat\alpha_{\rm PoS}=1$ and $\widehat\alpha_{\rm SL},\widehat\alpha_{\rm Aug},\widehat\alpha_{\rm PoS}\in [0,1]$. Here, $\widehat\theta_{jk, {\rm SCISS}}^{\rm Aug}$ and $\widehat\theta_{jk, {\rm SCISS}}^{\rm PoS}$ are SCISS estimators obtained with the methods in Section \ref{sec:Aug} and Section \ref{sec:PoS} respectively. The ES estimator derived in this way can achieve the (semiparametric) efficiency of the SCISS estimator with the correct model for $\by|\bz$ when either the augmented Ising model in \ref{sec:Aug} or the post-hoc surrogate model in Section \ref{sec:PoS} is correct.



%

\section{Simulation Studies}\label{sec:simulation}

We conducted comprehensive simulation studies to evaluate the performance of the proposed SCISS estimators and compare it with existing methods. Throughout, we set $N=10000$, $n=200$, $\bw=1$, and generate $q=3$ dimensional outcomes $\by$ from the Ising model (\ref{equ:Ising model}) with parameters
\[
\btheta=\begin{bmatrix}0.1&0.3&-0.6\\0.3&-0.3&0.4\\-0.6&0.4&0.2\end{bmatrix}.
\]
For the auxiliary features $\bx$, we set their dimension as $p=3$ and generate them based on $\by$ from two mechanisms including: (\romone) Gaussian model $x_k=c_{kk}y_k+\sum_{j\neq k}c_{jk}y_jy_k+\epsilon$ with $\epsilon\sim N(0,1)$; and (\romtwo) Poisson model $x_k\sim {\rm Poi}(c_{kk}y_k+\sum_{j\neq k}c_{jk}y_jy_k)$ for $k=1,\ldots,q$, where ${\rm Poi}(\lambda)$ refers to Poisson distribution with a rate $\lambda$. For the coefficient matrix $\c=(c_{jk})_{q\times q}$, we consider three choices:
\[
(1)~\c_1=\begin{bmatrix}3~~&0~~&0\\0~~&3~~&0\\0~~&0~~&3\end{bmatrix};\quad(2)~\c_2=\begin{bmatrix}2.5&0.2&0.5\\0.2&2.5&0.5\\0.5&0.5&2.5\end{bmatrix};\quad(3)~\c_3=\begin{bmatrix}1.5&1&1.5\\1&2&1\\1.5&1&1.5\end{bmatrix}.
\]
They represent the scenarios with null, moderate, and strong effects of the interaction terms $y_jy_k$ ($j\neq k$) encoded by each $x_j$ or $x_k$ respectively. From (1) to (3), the information in $x_j$ about the interactive modes among $\by$ gets intensified. When $\c$ is taken as $\c_1$ in setting (1), our (conditional) independence assumption for the post-hoc surrogate in Section \ref{sec:PoS} will strictly hold for $\bx$. This model assumption will be violated moderately for $\c_2$ and severely for $\c_3$. Meanwhile, compared to $\c_1$, each single $x_j$ generated under $\c_2$ and $\c_3$ will be more informative to the network structure of $\by$, which leads to a better fit on the augmented Ising model for $\by\sim\bz$ introduced in Section \ref{sec:Aug}.


Throughout, we use SCISS-Aug to denote the implementation of our method with conditional probability model proposed in Section \ref{sec:Aug} and SCISS-PoS for that in Section \ref{sec:PoS}. To compare our method with existing SL and SSL methods, we also include and study the SL estimator and the density ratio weighting method (DR) proposed by \cite{kawakita2014safe} in all settings. The implementation details of DR is provided in Appendix \ref{sec:dr}. Note that existing statistical SSL methods besides DR could not be directly used on graphical models. When implementing DR and SCISS-Aug for Poisson $\bx$, we take a $\log(x+1)$ transformation for preprocessing to make it less skewed. Table 1 presented below summarizes the estimation performance on the off-diagonal entries (pairwise conditional relationships) of $\btheta$, measured by bias, standard error (SE), relative efficiency (RE) to the SL estimator, as well as the coverage probability of our method in 95\% confidence interval (CI) estimation. Complete results including diagonal entries are presented in Appendix \ref{sec:complete simulation result}.

\begin{table}[!htbp]
\caption{Bias, standard error (SE), relative efficiency (RE) to SL, and coverage probability (CP) of 95\% confidence interval (CI) of the estimators for $\theta_{12}$, $\theta_{13}$, and $\theta_{23}$, under the settings introduced in Section \ref{sec:simulation}. The results are created based onn 500 replications for each configuration. The bold parts stand for the highest estimation efficiency among different methods.}
\small
\center
\begin{tabular}{|c|cc|ccc|cccc|cccc|}
\hline
{}&\multicolumn{2}{c|}{SL}& \multicolumn{3}{c|}{DR}&\multicolumn{4}{c|}{SCISS-Aug}&\multicolumn{4}{c|}{SCISS-PoS}\\\hline
{}&Bias&SE&Bias&SE&RE&Bias&SE&RE&CP&Bias&SE&RE&CP\\ \hline\hline
\multicolumn{14}{|c|}{(I) Gaussian $\bx$ \& (1) Null interaction $\c=\c_1$.}\\
\hline\hline
{$\theta_{12}$}& 0.002 & 0.294 & 0.002 & 0.295 & 1.00 & -0.005 & 0.242 & 1.49 & 0.96 & -0.014 & \pmb{0.208} & \pmb{2.01} & 0.95\\
{$\theta_{13}$}& -0.013 & 0.290 & -0.014 & 0.291 & 0.99 & -0.005 & 0.276 & 1.10 & 0.95 & 0.031 & \pmb{0.219} & \pmb{1.74} & 0.94\\
{$\theta_{23}$}& 0.003 & 0.289 & 0.002 & 0.291 & 0.99 & 0.007 & 0.288 & 1.01 & 0.96 & 0.006 & \pmb{0.212} & \pmb{1.86} & 0.95\\\hline\hline
\multicolumn{14}{|c|}{(I) Gaussian $\bx$ \& (2) Moderate interaction $\c=\c_2$.}\\\hline\hline
$\theta_{12}$& 0.002 & 0.294 & 0.003 & 0.295 & 1.00 & -0.003 & \pmb{0.236} & \pmb{1.55} & 0.95 & -0.013 & 0.237 & 1.55 & 0.95\\
$\theta_{13}$& -0.013 & 0.290 & -0.015 & 0.289 & 1.00 & -0.004 & 0.243 & 1.42 & 0.95 & 0.033 & \pmb{0.233} & \pmb{1.55} & 0.95\\
$\theta_{23}$& 0.003 & 0.289 & 0.002 & 0.291 & 0.99 & 0.009 & 0.262 & 1.22 & 0.96 & 0.004 & \pmb{0.223} & \pmb{1.68} & 0.95\\\hline\hline
\multicolumn{14}{|c|}{(I) Gaussian $\bx$ \& (3) Strong interaction $\c=\c_3$.}\\\hline\hline
$\theta_{12}$& 0.002 & 0.294 & 0.006 & 0.292 & 1.02 & 0.003 & \pmb{0.225} & \pmb{1.71} & 0.94 & -0.021 & 0.270 & 1.19 & 0.94\\
$\theta_{13}$& -0.013 & 0.290 & -0.015 & 0.280 & 1.07 & 0.001 & \pmb{0.223} & \pmb{1.68} & 0.95 & 0.032 & 0.271 & 1.15 & 0.94 \\
$\theta_{23}$& 0.003 & 0.289 & 0.003 & 0.288 & 1.01 & 0.010 & \pmb{0.235} & \pmb{1.51} & 0.96 & -0.001 & 0.269 & 1.16 & 0.94\\\hline\hline
\multicolumn{14}{|c|}{(II) Poisson $\bx$ \& (1) Null interaction $\c=\c_1$.}\\\hline\hline
$\theta_{12}$& 0.002 & 0.294 & 0.002 & 0.293 & 1.01 & 0.004 & 0.228 & 1.67 & 0.97 & 0.003 & \pmb{0.129} & \pmb{5.18} & 0.94\\
$\theta_{13}$& -0.013 & 0.290 & -0.014 & 0.291 & 0.99 & -0.007 & 0.279 & 1.08 & 0.96 & 0.005 & \pmb{0.134} & \pmb{4.66} & 0.93\\
$\theta_{23}$& 0.003 & 0.289 & 0.004 & 0.289 & 1.00 & 0.004 & 0.287 & 1.01 & 0.96 & -0.005 & \pmb{0.127}& \pmb{5.21} & 0.95\\\hline\hline
\multicolumn{14}{|c|}{(II) Poisson $\bx$ \& (2) Moderate interaction $\c=\c_2$.}\\\hline\hline
$\theta_{12}$& 0.002 & 0.294 & 0.005 & 0.296 & 0.99 & 0.003 & 0.222 & 1.76 & 0.96 & 0.003 & \pmb{0.141} & \pmb{4.38} & 0.94\\
$\theta_{13}$& -0.013 & 0.290 & -0.015 & 0.290 & 1.00 & -0.005 & 0.254 & 1.30 & 0.95 & 0.005 & \pmb{0.132} & \pmb{4.79} & 0.96\\
$\theta_{23}$& 0.003 & 0.289 & 0.004 & 0.290 & 1.00 & 0.011 & 0.280 & 1.07 & 0.97 & -0.010 & \pmb{0.144} & \pmb{4.04} & 0.94\\\hline\hline
\multicolumn{14}{|c|}{(II) Poisson $\bx$ \& (3) Strong interaction $\c=\c_3$.}\\\hline\hline
$\theta_{12}$& 0.002 & 0.294 & 0.007 & 0.296 & 0.99 & 0.005 & 0.209 & 1.99 & 0.96 & 0.006 & \pmb{0.173} &\pmb{2.89} & 0.94\\
$\theta_{13}$& -0.013 & 0.290 & -0.015 & 0.284 & 1.04 & -0.004 & 0.226 & 1.65 & 0.95 & 0.017 & \pmb{0.176} & \pmb{2.70} & 0.96\\
$\theta_{23}$& 0.003 & 0.289 & 0.006 & 0.289 & 1.00 & 0.010 & 0.254 & 1.30 & 0.96 & -0.012 & \pmb{0.187} & \pmb{2.39} & 0.92\\\hline
\end{tabular}
\end{table}

Across all configurations, the biases of all methods are much smaller than their SEs, and the two SCISS methods attain CPs close to the nominal level 95\%, which indicates the validity of our methods. Moreover, both versions of SCISS attain significantly better efficiency compared to SL and DR for all parameters and configurations, as the RE to SL of them is always above $1$ and higher than that of DR. For example, SCISS-PoS holds the RE higher than $1.5$ in five out of the six configurations.

Interestingly, under the Gaussian $\bx$ mechanism (\romone), the superiority gradually moves from SCISS-PoS to SCISS-Aug with the alternation of $\c$ from $\c_1$ to $\c_3$. For example, the RE (to SL) of SCISS-PoS on $\theta_{12}$ is around $0.5$ larger than that of SCISS-Aug under the setup (1) with $\c=\c_1$, equal to that of SCISS-Aug under (2) $\c=\c_2$, and $0.5$ smaller than SCISS-Aug under (3) $\c=\c_3$. When $\c=\c_1$, we have $x_j\perp \by_{\setminus j}| \{y_j,\bw\}$ and, thus, SCISS-PoS will correctly characterize $\by|\bz$ and achieves semiparametric efficiency. For the setup of $\c=\c_2$, small interactions of $\by$ on $\bz$ emerge so the generated data deviate moderately from the model assumption on SCISS-PoS introduced in Section \ref{sec:PoS}. Meanwhile, these interaction effects can be captured more effectively by the augmented Ising model for $\by|\bz$ in Section \ref{sec:Aug}, though its model (\ref{equ:Augmented-Ising}) may still be misspecified. This makes the relative performance of SCISS-Aug to SCISS-PoS better than the scenario with $\c=\c_1$. Furthermore, $\c_3$ encodes a stronger interaction effects in $\bz\mid \by$, which makes the modeling strategy of SCISS-Aug more effective than SCISS-PoS and, thus, the former having essentially smaller variance than the latter. Under the Poisson mechanism (\romtwo), SCISS-PoS consistently performs better than SCISS-Aug. This may be because that when generating Poisson $\bx$ with $\by$, certain information is lost from this non-linearity and discrete process, which could make SCISS-PoS using Poisson model simply better. Nonetheless, the gap between two methods still become smaller when changing $\c$ from $\c_1$ to $\c_3$ and this trend is consistent with the Gaussian setting (\romone).  These results suggest that it is worthwhile considering the intrinsic relationship between the $\by$ and $\bx$ when choosing between SCISS-Aug and SCISS-PoS. After all, our data-driven ensemble strategy in Section \ref{sec:ensemble} could typically help on this.


To investigate the validity and effectiveness of our proposed intrinsic efficient approach in Section\ref{sec:intrinsic}, we carry out an additional simulation study described in Appendix \ref{sec:complete simulation result} with details. The simulation setup include a data generation procedure on which SCISS-Aug and SCISS-PoS are highly misspecified, and is connected with the anchor-positive surrogate variable studied in the literature of EHR  \citep{zhang2020maximum}. With further intrinsic efficient update on SCISS-Aug, the resulted estimator $\widehat\theta_{jk,{\rm ES}}$ shows small enough bias and an $1.25$ relative efficiency to the original SCISS-Aug estimator, which means its performance is moderately better than the latter. In most other configurations, we found that $\widehat\theta_{jk,{\rm ES}}$ has a close performance to the original SCISS.

\section{Example: EHR Study of MIMIC-III}\label{sec:realdata}

On MIMIC-III data set, we apply the proposed method to infer the graphical model of several ICU-related phenotypes including {\em depression}, {\em obesity}, {\em alcohol abuse}, and {\em psychiatric disorder}. The whole study cohort consists of $N_0=38578$ patients who have been admitted to ICU. Among them, $N_1=1963$ were frequently ($\geq 3$ times) sent to ICU, referred as the FREQ group here, while the other $N_2=36615$ subjects form the non-FREQ group. The binary gold standard labels $\by$ for the $q=4$ phenotypes have been obtained on small subsets of subjects in both groups via chart-reviewing in \cite{moseley2020phenotype}. We thank the authors of this paper for providing this valuable information. In FREQ group, $n_1=436$ subjects are labeled and in non-FREQ, $n_2=609$ are labeled. For the auxiliary features, we take $\bx=(x_1,\ldots,x_4)\trans$ with each $x_j\in\{0,1\}$ as the indicator of having a non-zero count of the main ICD code for each phenotype $y_j$.

Next, we apply methods including SCISS, SL, and DR to estimate the Ising model for $\by$. For our method, we first extract SL, SCISS-Aug, and SCISS-PoS, and then ensemble them to produce the final estimator using the strategy described in Section \ref{sec:ensemble}. All methods are implemented on the FREQ and non-FREQ groups separately. In Table \ref{tab:real}, we present the point estimate, the empirical SE, and the RE to SL of the three methods on the FREQ and non-FREQ data. Over both SL and DR, our method shows no worse performance on all components, and achieves higher efficiency than on almost all of them. For example, on the marginal effect ($\theta_{jj}$) of alcohol abuse from the FREQ group, SCISS attains 67\% higher efficiency than SL while DR only shows an improvement of 13\%. For the conditional dependence effect between depression and obesity on non-FREQ, SCISS is 22\% more efficient than SL while DR and SL have nearly the same SEs. Note that in this real example, the empirical SEs of our method and DR will be always larger than or equal to that of the SL due to the use of the ES strategy in our method and the use of in-sample mean square errors to quantify the uncertainty of DR.

\begin{table}[!ht]
\caption{\label{tab:real} Point estimatie (Est), standard error (SE), and relative efficiency to SL (RE) of the SL, DR, and SCISS estimators on the FREQ and non-FREQ data sets in MIMIC-III. The bold part in the table stands for the highest RE among different methods. The diseases are abbreviated as DEPR (Depression), OB (Obesity), AA (Alcohol Abuse), PS (Psychiatric Disorder) respectively.}
\center
\begin{tabular}{|c|c|c|c|c|c|c|c|c|c|}\hline
{}&\multicolumn{2}{c|}{SL}& \multicolumn{3}{c|}{DR}&\multicolumn{3}{c|}{\rm SCISS}\\\hline 
{}&Est&SE&Est&SE&RE&Est&SE&RE\\\hline\hline
\multicolumn{9}{|c|}{FREQ}\\\hline\hline
{DEPR}& -1.35 & 0.14 & -1.37 & 0.14 & 1.05 & -1.32 & 0.13 & \pmb{1.20}\\
{DEPR $\leftrightarrow$ OB}& 0.40 & 0.24 & 0.55 & 0.23 & 1.10 & 0.58 & 0.22 & \pmb{1.16}\\
{DEPR $\leftrightarrow$ AA}& 0.24 & 0.21 & 0.44 & 0.21 & 1.00& 0.45 & 0.20 & \pmb{1.10}\\
{DEPR $\leftrightarrow$ PS}& 1.26 & 0.39 & 1.85 & 0.39 & 1.00 & 1.78 & 0.39 & \pmb{1.00}\\
{OB}& -2.78 & 0.25 & -2.19 & 0.21 & 1.42 & -2.24 & 0.20 & \pmb{1.58}\\
{OB $\leftrightarrow$ AA}& -0.78 & 0.22 & -1.21 & 0.21 & 1.04 & -1.30 & 0.19 & \pmb{1.25}\\
{OB $\leftrightarrow$ PS}& 0.62 & 0.25 & 0.42 & 0.24 & 1.03 & 0.55 & 0.23 & \pmb{1.16}\\
{AA}& -2.20 & 0.20 & -1.95 & 0.19 & 1.13 & -1.95 & 0.16 & \pmb{1.67}\\
{AA $\leftrightarrow$ PS}& 0.75 & 0.23 & 0.33 & 0.23 & 1.01 & 0.38 & 0.22 &\pmb{1.11}\\
{PS}& -2.05 & 0.27 & -2.36 & 0.25 & 1.17 & -2.42 & 0.24 & \pmb{1.18}\\
\hline\hline
\multicolumn{9}{|c|}{non-FREQ}\\\hline\hline
{DEPR}& -1.35 & 0.12 & -1.69 & 0.10 & 1.36 & -1.69 & 0.10 & \pmb{1.40}\\
{DEPR $\leftrightarrow$ OB}& 0.40 & 0.21 & 0.22 & 0.21 & 1.04& 0.44 & 0.19 & \pmb{1.22}\\
{DEPR $\leftrightarrow$ AA}& 0.24 & 0.19 & 0.18 & 0.19 & 1.01 & 0.37 & 0.18 & \pmb{1.09}\\
{DEPR $\leftrightarrow$ PS}& 1.26 & 0.26 & 1.56 & 0.25 & 1.02 & 1.42 & 0.25 & \pmb{1.04}\\
{OB}& -2.78 & 0.26 & -2.90 & 0.22 & 1.43 & -3.00 & 0.21 & \pmb{1.60}\\
{OB $\leftrightarrow$ AA}& -0.78 & 0.19 & -0.36 & 0.17 & 1.17& -0.65 & 0.17 & \pmb{1.19}\\
{OB $\leftrightarrow$ PS}& 0.62 & 0.22 & 0.82 & 0.22 & 1.06 & 0.76 & 0.21 & \pmb{1.07}\\
{AA}& -2.20 & 0.19 & -2.26 & 0.17 & 1.27 & -2.33 & 0.16 & \pmb{1.40}\\
{AA $\leftrightarrow$ PS}& 0.75 & 0.22 & 0.73 & 0.22 & 1.01 & 0.79 & 0.21 & \pmb{1.10}\\
{PS}& -2.05 & 0.18 & -2.35 & 0.16 & 1.31 & -2.39 & 0.15 & \pmb{1.45}\\\hline
\end{tabular}
\end{table}

In Figure \ref{fig:real data}, we present two diagrams sketching SCISS fitted graphs (including 95\% CI) of the four phenotypes on the FREQ and non-FREQ subjects separately. Depression and psychiatric disorder shows the most significant positive conditional dependence on both FREQ and non-FREQ, which is in line with existing literature in the application field. For example, in the study of \cite{1990Adult}, the depressed adolescent group showed an increased risk for psychiatric disorder in adult life as well as elevated risks of psychiatric hospitalization and treatment. The strongest negative dependence comes between alcohol abuse and obesity, which could be explained by the intensive metabolism after ethanol take-in. According to \cite{addolorato1998influence}, alcoholics show a significantly lower body weight and a significantly lower fat mass compared with their controls. Interestingly, the conditional dependence effects of alcohol abuse and obesity are significantly different between the FREQ and non-FREQ groups, as marked in red in Figure \ref{fig:real data}. The $p$-value of the two-sample test on this edge between the two groups is $0.006$, which is the only one being significant contrasting the two groups. This distinction may be because that patients frequently enrolling ICU are more likely to be severely bothered by alcohol abuse, which enhances its effect on the loss of body weight and reduces the risk of obesity.

\begin{figure}[htpb!]
    \caption{The point estimate and the 95\% confidence interval by SCISS of the Ising parameters for the FREQ and non-FREQ groups.}
    \centering
    \includegraphics{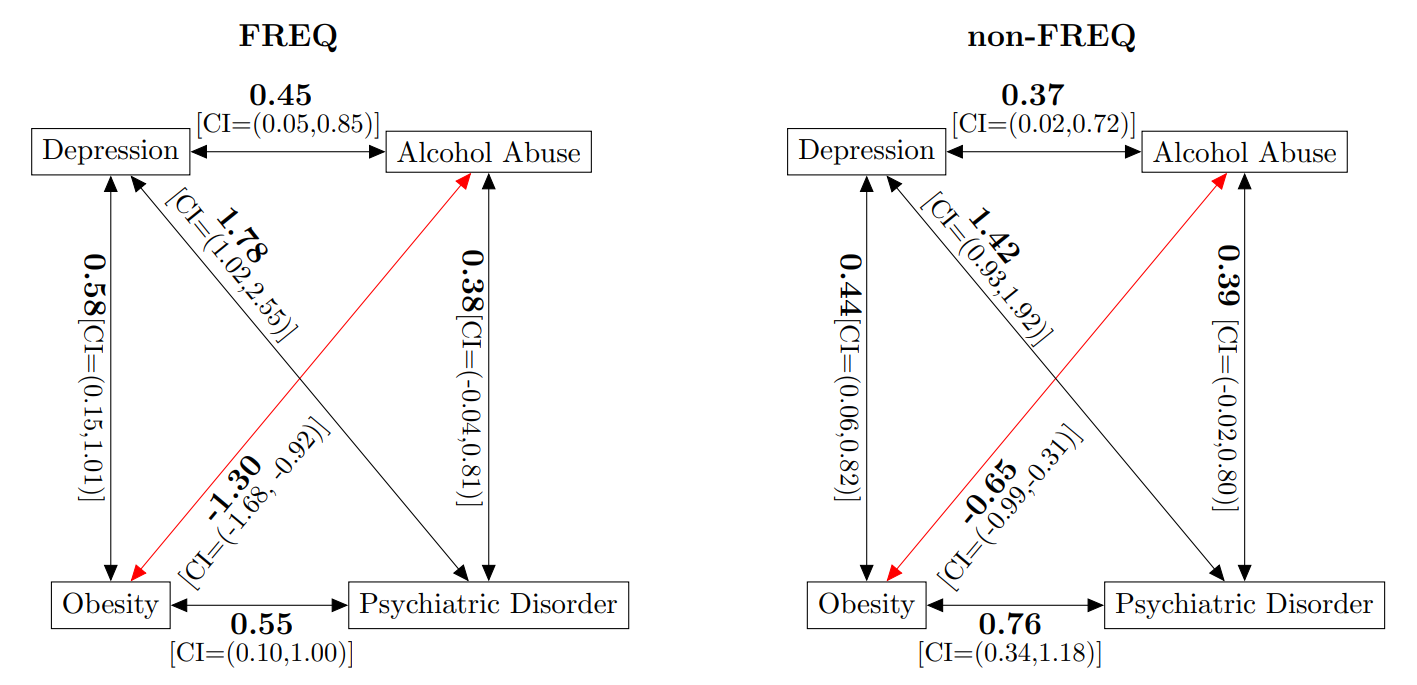}
    \label{fig:real data}
\end{figure}

\section{Discussion}

In this paper, we focus on the efficient estimation of Ising models under a common SSL setting with auxiliary features. We propose a novel SSL framework, under which the augmented Ising model (Aug) and post-hoc surrogate model (PoS) are introduced to leverage the useful information about $\by$ from the large sample of $\bx$. In addition, we develop two approaches including the intrinsic efficient SSL and the ensemble estimation, to enhance the performance of our estimator under potentially misspecified models for $\by\sim\bx$. Through theoretical studies, we demonstrate the advantage of the SCISS estimator over the SL estimator in terms of efficiency, as well as the effectiveness of our proposal in overcoming the misspecification issue. Our numerical studies indicate that both our SCISS-Aug and SCISS-PoS realizations outperform existing SL and SSL methods. 


We shall point out several potential directions to enable a more comprehensive and effective use of our framework. First, SCISS can be naturally extended to address other classes of exponential graph models with different types of outcomes $\by$, such as Gaussian graphical models for continuous $\by$ and Poisson graphs for counting outcomes. This is because these graphical models are also estimated through conditional regression like (\ref{equ:SL}), to which our score-imputation framework and specific modeling strategies can be generalized well. One can further consider the SSL problem of the semiparametric exponential graphs introduced by \cite{yang2018semiparametric}, with each $\theta_{jk}$ estimated by maximizing a nuisance-free pseudo-likelihood based on pairwise comparison between the subjects. In this case, robust and efficient imputation of their score functions formed as U-statistics is an open problem warranting future research. 
 
Second, it is appealing to utilize SCISS to address high-dimensional outcomes $\by$ and covariates $\bz$. In specific, our method can be incorporated with recent advancements in debiased inference \citep[e.g.]{zhang2014confidence,van2014asymptotically,javanmard2019false} for efficient SS learning and inference of high-dimensional graphs, enabling application to phenome-wide studies with a large set of diseases and more elements like genetic modification of the network. The idea is to induce the score function of the debiased estimator for each $\theta_{jk}$. Note that our construction in (\ref{equ:construction}) is automatically correcting the estimation error of $\widehat\m_j(\bz^i;\widehat\bfeta,\widehat\btheta_{j,{\rm SL}})$ that could be excessive in the high-dimensional regime. Nevertheless, the computation burden of extracting $\widehat\m_j(\bz^i;\widehat\bfeta,\widehat\btheta_{j,{\rm SL}})$ may be a major concern in the scenario. In addition, the group inference of a sub-network considered in \cite{xia2018multiple} requires aggregating the estimators of multiple $\theta_{jk}$'s, which will present a new challenge for our SSL framework.

\bibliographystyle{apalike}
\bibliography{SSLref} 
\newpage

\newpage
\appendix

\section*{Appendix}
\numberwithin{equation}{section}
\setcounter{lemma}{0}
\setcounter{cond}{0}
\setcounter{table}{0}
\setcounter{figure}{0}
\setcounter{remark}{0}
\setcounter{theorem}{0}

\renewcommand{\thelemma}{A.\arabic{lemma}}
\renewcommand{\thetheorem}{A.\arabic{theorem}}
\renewcommand{\thecond}{A.\arabic{cond}}
\renewcommand{\thefigure}{A.\arabic{figure}}
\renewcommand{\thetable}{A.\arabic{table}}
\renewcommand{\theremark}{A.\arabic{remark}}
Here we provide justifications for our main theoretical results. 
\section{Asymptotic Properties of $\widehat\btheta_{\rm SL}$, $\widehat\btheta_{\rm SCISS} $ and $\widehat\btheta_{\rm INTR}$}\label{sec:proof_SL}
In this section, we will prove Lemma \ref{lemma:SL} presented in Section \ref{sec:method:sup}, Theorem \ref{thm:SSL} and Theorem \ref{thm:intri} in Section \ref{sec:method}. We first introduce the following conditions, as they are commonly used regularity conditions for the continuously differentiability of variables and the positive definition of hessian matrix in M-estimation theory and are satisfied in broad applications. Note that we let $\Sigma_1\succ\Sigma_2$ if $\Sigma_1-\Sigma_2$ is positive definite and $\Sigma_1\succeq\Sigma_2$ if $\Sigma_1-\Sigma_2$ is positive semi-definite for any two symmetric matrices $\Sigma_1$ and $\Sigma_2$. 

\begin{cond}
  $\bz=\{\bx,\bw\}$ have compact support and the density function for $\bz$ is continuously differentiable in the continuous components of $\bz$.
\label{cond:1}
\end{cond}

\begin{cond} For $j=1,\ldots,q,$\\
{\rm(A)} The true hessian matrix for Ising estimating equation $\Sigma_{\overline\btheta_j}\succ\bzero,$\\ 
{\rm(B)} The true hessian matrix for estimating equation in Section \ref{sec:Aug}, ${\rm E}[\bpsi_{\setminus j,\bw}\bpsi_{\setminus j,\bw}\trans\dot g(\overline\bfeta_j \trans\bpsi_{\setminus j,\bw})]\succ\bzero,$\\
{\rm(C)} The true hessian matrix for estimating equation in Section \ref{sec:PoS}, {\rm (1)}$ {\rm E}[(\bw\trans,y_j)\trans(\bw\trans,y_j)]\succ\bzero,$ {\rm (2)} ${\rm E}[(\bw\trans,y_j)\trans(\bw\trans,y_j)(x_j-1)]\succ\bzero,$ {\rm (3)} ${\rm E}[(\bw\trans,y_j)\trans(\bw\trans,y_j)\exp\{(\bw\trans,y_j)\bbeta_j\}]\succ\bzero.$
\label{cond:2} 
\end{cond}

\begin{cond} For $j,k=1,\ldots,q,$ the hessian matrix for intrinsic estimating equation in Section \ref{sec:intrinsic},
\begin{equation*}
\begin{aligned}
      {\rm E}\bigg(&\frac{\partial(m_{jk}(\bz;\overline\bfeta_{\rm INTR},\overline\btheta_j))}{\partial \bfeta}\cdot\frac{\partial(m_{jk}(\bz;\overline\bfeta_{\rm INTR},\overline\btheta_j))}{\partial\bfeta\trans}\\
      &-\{s_{jk}(\by_{\bw};\overline\btheta_j)-m_{jk}(\bz;\overline\bfeta_{\rm INTR},\overline\btheta_j)\}\cdot\frac{\partial^2\{m_{jk}(\bz;\overline\bfeta_{\rm INTR},\overline\btheta_j)\}}{\partial\bfeta\partial\bfeta\trans}\\
      &+\frac{\partial(m_{kj}(\bz;\overline\bfeta_{\rm INTR},\overline\btheta_k))}{\partial \bfeta}\cdot\frac{\partial(m_{kj}(\bz;\overline\bfeta_{\rm INTR},\overline\btheta_k))}{\partial\bfeta\trans}\\
      &-\{s_{kj}(\by_{\bw};\overline\btheta_k)-m_{jk}(\bz;\overline\bfeta_{\rm INTR},\overline\btheta_j)\}\cdot\frac{\partial^2\{m_{kj}(\bz;\overline\bfeta_{\rm INTR},\overline\btheta_k)\}}{\partial\bfeta\partial\bfeta\trans}\bigg)\succ\bzero.
\end{aligned}
\end{equation*}\label{cond:intri}
\end{cond}

We start from the proof of Lemma \ref{lemma:SL}, i.e. the asymptotic property of $\widehat\btheta_{\rm SL}$. For the simplicity in proof, we denote the true realization of estimating equation (\ref{equ:SL}) as 
\begin{equation*}
    \mathbf{U}_j(\btheta_j)={\rm E} [\by_{\setminus j,\bw}\{y_j-g(\btheta_j\trans\by_{\setminus j,\bw})\}]=\bzero,{~}j=1,\ldots,q.
\end{equation*}

We begin by verifying that $\widecheck\btheta_{j,{\rm SL}}$ is consistent for $\overline\btheta_j.$ It suffices to show that (i) sup$_{\btheta_j\in \Theta_j}\Vert\mathbf{U}_{j,n}(\btheta_j)-\mathbf{U}_{j}(\btheta_j)\Vert_2=o_p(1)$ and 
(ii) inf$_{\Vert\btheta_j-\overline\btheta_j\Vert_2>\epsilon}\Vert\mathbf{U}_j(\btheta_j)\Vert_2>0$ for any $\epsilon>0$ from \cite{van2000asymptotic}. 
Since $\by \in \{0,1\}^q$ and $g(\cdot)$ is bounded, we have that 
$n^{-1}\sum_{i=1}^n\by^i_{\setminus j,\bw}\{y_j^i-g(\btheta_j\trans\by_{\setminus j,\bw}^i)\}$ converges to ${\rm E}[\by_{\setminus j,\bw}\{y_j-g(\btheta_j\trans\by_{\setminus j,\bw})\}]$ in probability uniformly as $n\xrightarrow{}\infty$ from the uniform law of large number (ULLN). It follows that sup$_{\btheta_j\in\Theta}\Vert\mathbf{U}_{j,n}(\btheta_j)-\mathbf{U}_{j}(\btheta_j)\Vert_2=o_p(1),$ then (i) is satisfied. Furthermore, (ii) follows directly from
Condition \ref{cond:2} (A) that $\mathbf U_j(\btheta_j)$ is monotonous to $\btheta_j$, then unique $\overline \btheta_j$ exists. Consequently, $\widecheck\btheta_{j,{\rm SL}}\xrightarrow{p}\overline\btheta_j$ as $n\xrightarrow{}\infty$.
Next we consider the asymptotic normality of $n^{\frac{1}{2}}(\widecheck\btheta_{j,{\rm SL}}-\overline\btheta_j).$ Noting that $\widecheck\btheta_{j,{\rm SL}}\xrightarrow{p}\overline\btheta_j,$ it is well-known that when $n\xrightarrow{}\infty$,
$$
n^{\frac{1}{2}}(\widecheck\btheta_{j,{\rm SL}}-\overline\btheta_j)=n^{-\frac{1}{2}}\sum_{i=1}^n\S_j(\by_{\bw}^i;\overline\btheta_j)+o_p(1).
$$
By symmetrization described in Section \ref{sec:method:sup}, we obtain the Taylor expansion for $\widehat\theta_{jk,{\rm SL}}$ as
$$
n^{\frac{1}{2}}(\widehat\theta_{jk,{\rm SL}}-\overline\theta_{jk})=\frac{1}{2}n^{-\frac{1}{2}}\sum_{i=1}^n\{s_{jk}(\by_{\bw}^i;\overline\btheta_j)+s_{kj}(\by_{\bw}^i;\overline\btheta_k)\}+o_p(1).
$$
Under Condition $\ref{cond:2}$ (A), we have finished the proof of Lemma \ref{lemma:SL} by the classical Central Limit Theorem that 
\[n^{\frac{1}{2}}(\widehat\theta_{jk,{\rm SL}}-\overline\theta_{jk})\xrightarrow{}N(0,\Omega_{jk,{\rm SL}}), {\rm ~where~} \Omega_{jk,{\rm SL}}=\frac{1}{4}{\rm E}\{s_{jk}(\by_{\bw};\overline\btheta_j)+s_{kj}(\by_{\bw};\overline\btheta_k)\}^2.\]

\section{Asymptotic Properties of $\widehat\btheta_{\rm SCISS}$}\label{sec:proof_SSL}
In this section, we will provide the proof for our main Theorem \ref{thm:SSL}.

We first present the asymptotic property of $\widehat\bfeta$ we use in imputing the score function and start from introducing the true value of $\bfeta$ in Section \ref{sec:Aug} and \ref{sec:PoS}. We now denote the true value of $\widehat\bfeta$ in Section \ref{sec:Aug} as $\overline\bfeta_{\rm A}$, which is the solution to estimating equation
\begin{equation*}
    \mathbf{Q}^{\rm Aug}_j(\bfeta_{j})
    ={\rm E} [\boldsymbol{\psi}_{\setminus j,\bw}\{y_j-g(\bfeta_{j}\trans\boldsymbol{\psi}_{\setminus j,\bw})\}]=\bzero, {~}j=1,\ldots,q.
\end{equation*}
    
Meanwhile, we denote the true value for $\bfeta$ of different assumption models in Section \ref{sec:PoS} as $\overline\bfeta_{{P-\rm I}}$, $\overline\bfeta_{{P-\rm II}}$ and $\overline\bfeta_{{P-\rm III}}$, as they are solutions to estimating equation

(I) $\mathbf{Q}^{PoS-{\rm I}}_j(\bfeta_j)={\rm E}[(\bw,y_j)\trans\{x_j-(\bw\trans,y_j)\bfeta_j\}]=\bzero$, 

(II) $\mathbf{Q}^{PoS-{\rm II}}_j(\bfeta_j)={\rm E}[(\bw,y_j)\trans\{x_j-{\rm expit}(\bw\trans,y_j)\bfeta_j\}]=\bzero$, 

(III) $\mathbf{Q}^{PoS-{\rm III}}_j(\bfeta_j)={\rm E}[(\bw,y_j)\trans\{x_j-{\rm exp}(\bw\trans,y_j)\bfeta_j\}]=\bzero,$ for $j=1,\ldots,q$, respectively.

Next, we will present the asymptotic property of the true value of $\bfeta$ above in Lemma \ref{lemma:eta}, as for simplicity, we unify the notation of them as $\overline\bfeta$ and will maintain the use of $\overline\bfeta$ throughout the proof of the Theorem \ref{thm:SSL}.

\begin{lemma}\label{lemma:eta}
Under Condition \ref{cond:1} and Condition \ref{cond:2} {\rm (B),(C)}, we have $n^{\frac{1}{2}}(\widehat\bfeta-\overline\bfeta)=O_p(1).$
\end{lemma}
\begin{proof}
~\\
For $\widehat\bfeta_A$, under Condition $\ref{cond:1}$ and Condition $\ref{cond:2}$ (B), considering that $\lambda_n=o(n^{-\frac{1}{2}})$, we can adapt the same procedure of proving Lemma \ref{lemma:SL} in Appendix \ref{sec:proof_SL} to show that $\widehat\bfeta_A\xrightarrow{p}\overline\bfeta_A$ as $n\xrightarrow{}\infty$. We then can obtain the following Taylor expansion 
$$
n^{\frac{1}{2}}(\widehat\bfeta_{A,j}-\overline\bfeta_{A,j})=n^{-\frac{1}{2}}\sum_{i=1}^n[{\rm E}\{\bpsi_{\setminus j,\bw}\bpsi_{\setminus j,\bw}\trans\dot g(\overline\bfeta_{A,j}\trans\bpsi_{\setminus j,\bw})\}]^{-1}\bpsi_{\setminus j,\bw}^i\{y_j^i-g(\overline\bfeta_{A,j}\trans\bpsi_{\setminus j,
\bw}^i)\}+o_p(1).
$$
It follows that $n^{\frac{1}{2}}(\widehat\bfeta_{A,j}-\overline\bfeta_{A,j})$ converges to a zero-mean Gaussian distribution thus $n^{\frac{1}{2}}(\widehat\bfeta_{A,j}-\overline\bfeta_{A,j})=O_p(1)$. Then we have $n^{\frac{1}{2}}(\widehat\bfeta_A-\overline\bfeta_A)=O_p(1).$

For $\widehat\bfeta_{P-{\rm I}}$, $\widehat\bfeta_{P-{\rm II}}$ and $\widehat\bfeta_{P-{\rm III}}$, under Condition \ref{cond:1} and Condition \ref{cond:2} (C), we also have that the hessian matrix for $\mathbf{Q}_j^{PoS-{\rm I}}(\bfeta_j),$ $\mathbf{Q}_j^{PoS-{\rm II}}(\bfeta_j){\rm ~and~}\mathbf{Q}_j^{PoS-{\rm \romthree}}(\bfeta_j)$ are positive. It then follows from the same procedure for $\widehat\bfeta_A$ that $n^{\frac{1}{2}}(\widehat\bfeta_{P}-\overline\bfeta_{P})=O_p(1),$ where $\widehat\bfeta_P$ including $\widehat\bfeta_{P-{\rm I}}$, $\widehat\bfeta_{P-{\rm II}}$ and $\widehat\bfeta_{P-{\rm III}}$.

Above all, we have $n^{\frac{1}{2}}(\widehat\bfeta-\overline\bfeta)=O_p(1)$.
\end{proof}

We begin to derive the asymptotic property of $\widehat\btheta_{\rm SCISS}$ now. Firstly, the construction of $\widehat\btheta_{\rm SCISS}$ yields the expansion for $n^{\frac{1}{2}}(\widehat\theta_{jk,\rm SCISS}-\widehat\btheta_{jk,\rm SCISS})$ as
\begin{equation}\label{equ:SCISS_proof_first}
\begin{aligned}
n^{\frac{1}{2}}(\widehat{\theta}_{jk,{\rm SCISS}}-\overline\theta_{jk})=&\frac{1}{2}n^{-\frac{1}{2}}\sum_{i=1}^n \{s_{jk}(\by_{\bw}^i;\overline\btheta_j)+s_{kj}(\by_{\bw}^i;\overline\btheta_k)\}-\frac{1}{2}n^{-\frac{1}{2}}\sum_{i=1}^{n}\{\widehat{m}_{jk}(\bz^i;\widehat\bfeta,\widehat\btheta_{j,{\rm SL}})+\widehat{m}_{kj}(\bz^i;\widehat\bfeta,\widehat\btheta_{k,{\rm SL}})\}\\
&+\frac{1}{2}n^{\frac{1}{2}}N^{-1}\sum_{i=n+1}^{n+N} \{\widehat{m}_{jk}(\bz^i;\widehat\bfeta,\widehat\btheta_{j,{\rm SL}})+\widehat{m}_{kj}(\bz^i;\widehat\bfeta,\widehat\btheta_{k,{\rm SL}})\}+o_p(1).
\end{aligned}
\end{equation}

Note that since $\widehat{m}_{jk}(\bz;\bfeta,\btheta_j)$ is differential to $\bfeta$ and $\widehat\bfeta\xrightarrow{p}\overline\bfeta,$ we can Taylor expand $\widehat{m}_{jk}(\bz;\widehat\bfeta,\widehat\btheta_j)$ at $\overline\bfeta$ and have
\begin{equation*}
  \widehat{m}_{jk}(\bz;\widehat\bfeta,\widehat\btheta_{j,{\rm SL}})=\widehat{m}_{jk}(\bz;\overline\bfeta,\widehat\btheta_{j,{\rm SL}})+(\widehat\bfeta-\overline\bfeta)\trans\frac{\partial \widehat{m}_{jk}(\bz;\bfeta^*,\widehat\btheta_{j,{\rm SL}})}{\partial\bfeta},
\end{equation*} 
where $\bfeta^*$ is a vector between $\widehat\bfeta$ and $\overline\bfeta$.

After obtaining the expansion of $\widehat m_{jk}$, we now are able to write equation (\ref{equ:SCISS_proof_first}) as
\begin{equation}
\begin{aligned}
    n^{\frac{1}{2}}(\widehat{\theta}_{jk,{\rm SCISS}}-\overline\theta_{jk})=&\frac{1}{2}n^{-\frac{1}{2}}\sum_{i=1}^n\{s_{jk}(\by_{\bw}^i;\overline\btheta_j)+s_{kj}(\by_{\bw}^i;\overline\btheta_k)\}\\
    &-\frac{1}{2}n^{-\frac{1}{2}}\sum_{i=1}^n\{\widehat{m}_{jk}(\bz^i;\overline\bfeta,\widehat\btheta_{j,{\rm SL}})+(\widehat\bfeta-\overline\bfeta)\trans\frac{\partial \widehat{m}_{jk}(\bz^i;\bfeta^{*1},\widehat\btheta_{j,{\rm SL}})}{\partial \bfeta}\\
    &+\widehat{m}_{kj}(\bz^i;\overline\bfeta,\widehat\btheta_{k,{\rm SL}})+(\widehat\bfeta-\overline\bfeta)\trans\frac{\partial \widehat{m}_{kj}(\bz^i;\bfeta^{*2},\widehat\btheta_{k,{\rm SL}})}{\partial \bfeta}\}\\
    &+\frac{1}{2}n^{\frac{1}{2}}N^{-1}\sum_{i=n+1}^{n+N}\{\widehat{m}_{jk}(\bz^i;\overline\bfeta,\widehat\btheta_{j,{\rm SL}})+(\widehat\bfeta-\overline\bfeta)\trans\frac{\partial \widehat{m}_{jk}(\bz^i;\bfeta^{*3},\widehat\btheta_{j,{\rm SL}})}{\partial \bfeta}\\
    &+\widehat{m}_{kj}(\bz^i;\overline\bfeta,\widehat\btheta_{k,{\rm SL}})+(\widehat\bfeta-\overline\bfeta)\trans\frac{\partial \widehat{m}_{kj}(\bz^i;\bfeta^{*4},\widehat\btheta_{k,{\rm SL}})}{\partial \bfeta}\}+o_p(1),
\label{equ:expansion_1}
\end{aligned}
\end{equation}
where $\bfeta^{*i},i=1,2,3,4$ are vectors between $\widehat\bfeta$ and $\overline\bfeta$.

For the clarity in further analysis, we split the expansion (\ref{equ:expansion_1}) into two parts, as 
\begin{equation*}
\begin{aligned}
    n^{\frac{1}{2}}(\widehat{\theta}_{jk,{\rm SCISS}}-\overline\theta_{jk})=A-B+o_p(1),
    \label{SSL_expansion}
\end{aligned}
\end{equation*}
where 
\begin{equation*}
\begin{aligned}
A=&\frac{1}{2}n^{-\frac{1}{2}}\sum_{i=1}^n\{s_{jk}(\by_{\bw}^i;\overline\btheta_j)+s_{kj}(\by_{\bw}^i;\overline\btheta_k)-\widehat{m}_{jk}(\bz^i;\overline\bfeta,\widehat\btheta_{j,{\rm SL}})-\widehat{m}_{kj}(\bz^i;\overline\bfeta,\widehat\btheta_{k,SL})\}\\
&+\frac{1}{2}n^{\frac{1}{2}}N^{-1}\sum_{i=n+1}^{n+N}\{\widehat{m}_{jk}(\bz^i;\overline\bfeta,\widehat\btheta_{j,{\rm SL}})+\widehat{m}_{kj}(\bz^i;\overline\bfeta,\widehat\btheta_{k,SL})\}   
\end{aligned}
\label{equ:A}
\end{equation*}    
and
\begin{equation*}
\begin{aligned}
B=&(\widehat\bfeta-\overline\bfeta)\trans[\frac{1}{2}n^{-\frac{1}{2}}\sum_{i=1}^n\{\frac{\partial \widehat{m}_{jk}(\bz^i;\bfeta^{*1},\widehat\btheta_{j,{\rm SL}})}{\partial \bfeta}+\frac{\partial \widehat{m}_{kj}(\bz^i;\bfeta^{*2},\widehat\btheta_{k,SL})}{\partial \bfeta}\}\\
&-\frac{1}{2}n^{\frac{1}{2}}N^{-1}\sum_{i=n+1}^{n+N}\{\frac{\partial \widehat{m}_{jk}(\bz^i;\bfeta^{*3},\widehat\btheta_{j,{\rm SL}})}{\partial \bfeta}+\frac{\partial \widehat{m}_{kj}(\bz^i;\bfeta^{*4},\widehat\btheta_{k,SL})}{\partial \bfeta}\}].
\end{aligned}
\label{equ:B}
\end{equation*}

We first analysis the property of A.

Since $\widehat{m}_{jk}(\bz;\bfeta,\btheta_j)$ is differential to $\btheta_j$ and $\widehat\btheta_{j,{\rm SL}}\xrightarrow{p}\overline\btheta_j$, we can Taylor expand $\widehat{m}_{jk}(\bz;\overline\bfeta,\widehat\btheta_{j,{\rm SL}})$ at $\overline\btheta_j$ as
\begin{equation}
\begin{aligned}
    \widehat {m}_{jk}(\bz;\overline\bfeta,\widehat\btheta_{j,{\rm SL}})=\widehat{m}_{jk}(\bz;\overline\bfeta,\overline\btheta_j)+(\widehat\btheta_{j,{\rm SL}}-\overline\btheta_j)\trans\frac{\partial \widehat{m}_{jk}(\bz;\overline\bfeta,\btheta_j^*)}{\partial \btheta_j},
    \label{equ:expand_theta}
\end{aligned}
\end{equation}
where $\btheta_j^*$ is a vector between $\widehat\btheta_{j,{\rm SL}}$ and thus $\overline\btheta_j$ and $\btheta_j^*\xrightarrow{p}\overline\btheta_j$ as well. 

To further analysis $\widehat m_{jk}(\bz;\overline\bfeta,\widehat\btheta_{j,{\rm SL}})$ and $\partial \widehat m_{jk}(\bz;\overline\bfeta,\btheta_j^*)/\partial \btheta_j$, we recall the estimate of hessian matrix $\widehat\Sigma_{\btheta_j}$ in the construction of $\widehat m_{jk}$. Since $\btheta_j^*\xrightarrow{p}\overline\btheta_j$, we have 
\begin{equation*}
    \widehat\Sigma_{\btheta_j^*}=n^{-1}\sum_{i=1}^n[\by_{\setminus j,\bw}^i(\by_{\setminus j,\bw}^i)\trans \dot g\{(\btheta_j^*)\trans \by_{\setminus j,\bw}^i\}]=n^{-1}\sum_{i=1}^n \{\by_{\setminus j,\bw}^i(\by_{\setminus j,\bw}^i)\trans \dot g(\overline\btheta_j\trans \by_{\setminus j,\bw}^i)+o_p(1)\}=\Sigma_{\overline\btheta_j}+o_p(1)
\end{equation*}
by the continuous mapping theorem and the weak law of large number when $n\xrightarrow{}\infty$.
Then from $\btheta_j^*\xrightarrow{p}\overline\btheta_j$ and $n^{\frac{1}{2}}(\widehat\btheta_{j,{\rm SL}}-\overline\btheta_j)=O_p(1)$ proofed in Appendix \ref{sec:proof_SL}, we have the following expansion of $\widehat m_{jk}(\bz;\overline\bfeta,\widehat\btheta_{j,{\rm SL}})$ from (\ref{equ:expand_theta}) as  
\begin{equation*}
\begin{aligned}
\widehat {m}_{jk}(\bz;\overline\bfeta,\widehat\btheta_{j,{\rm SL}})=&\widehat m_{jk}(\bz;\overline\bfeta,\overline\btheta_j)+(\widehat\btheta_{j,{\rm SL}}-\overline\btheta_j)\trans\{\frac{\partial m_{jk}(\bz;\overline\bfeta,\overline\btheta_j)}{\partial \btheta_j}+o_p(1)\}\\
=&m_{jk}(\bz;\overline\bfeta,\overline\btheta_j)+(\widehat\btheta_{j,{\rm SL}}-\overline\btheta_j)\trans\frac{\partial m_{jk}(\bz;\overline\bfeta,\overline\btheta_j)}{\partial \btheta_j}+o_p(n^{-\frac{1}{2}}),
\end{aligned}
\end{equation*}

Then we have 
\begin{equation}
    \begin{aligned}
A=&\frac{1}{2}n^{-\frac{1}{2}}\sum_{i=1}^n[s_{jk}(\by_{\bw}^i;\overline\btheta_j)+s_{kj}(\by_{\bw}^i;\overline\btheta_k)-\{m_{jk}(\bz^i;\overline\bfeta,\overline\btheta_j)+(\widehat\btheta_{j,{\rm SL}}-\overline\btheta_j)\trans\frac{\partial {m}_{jk}(\bz^i;\overline\bfeta,\overline\btheta_j)}{\partial \btheta_j}\\
&+m_{kj}(\bz^i;\overline\bfeta,\overline\btheta_k)+(\widehat\btheta_{k,SL}-\overline\btheta_k)\trans\frac{\partial {m}_{kj}(\bz^i;\overline\bfeta,\overline\btheta_k)}{\partial \btheta_k}\}+o_p(n^{-\frac{1}{2}})]\\
&+\frac{1}{2}n^{\frac{1}{2}}N^{-1}\sum_{i=n+1}^{n+N}\{m_{jk}(\bz^i;\overline\bfeta,\overline\btheta_j)+(\widehat\btheta_{j,{\rm SL}}-\overline\btheta_j)\trans\frac{\partial {m}_{jk}(\bz^i;\overline\bfeta,\overline\btheta_j)}{\partial \btheta_j}\\
&+m_{kj}(\bz^i;\overline\bfeta,\overline\btheta_k)+(\widehat\btheta_{k,SL}-\overline\btheta_k)\trans\frac{\partial {m}_{kj}(\bz^i;\overline\bfeta,\overline\btheta_k)}{\partial \btheta_k}+o_p(n^{-\frac{1}{2}})\},
\end{aligned}
\end{equation}
which can be written as 
\begin{equation*}
A=A_1-A_2,
\end{equation*}
where 
\begin{equation*}
\begin{aligned}
  A_1=&\frac{1}{2}n^{-\frac{1}{2}}\sum_{i=1}^n\{s_{jk}(\by_{\bw}^i;\overline\btheta_j)+s_{kj}(\by_{\bw}^i;\overline\btheta_k)-m_{jk}(\bz^i;\overline\bfeta,\overline\btheta_j)-m_{kj}(\bz^i;\overline\bfeta,\overline\btheta_k)\}\\
    &+\frac{1}{2}n^{\frac{1}{2}}N^{-1}\sum_{i=n+1}^{n+N}\{m_{jk}(\bz^i;\overline\bfeta,\overline\btheta_j)+m_{kj}(\bz^i;\overline\bfeta,\overline\btheta_k)\}+o_p(1)  
\end{aligned}
\end{equation*}
and
\begin{equation*}
    \begin{aligned}
        A_2=&\frac{1}{2}n^{-\frac{1}{2}}\sum_{i=1}^n\{(\widehat\btheta_{j,{\rm SL}}-\overline\btheta_j)\trans\frac{\partial m_{jk}(\bz^i;\overline\bfeta,\overline\btheta_j)}{\partial \btheta_j}+(\widehat\btheta_{k,SL}-\overline\btheta_k)\trans\frac{\partial m_{kj}(\bz^i;\overline\bfeta,\overline\btheta_k)}{\partial \btheta_k})\}\\
        &-\frac{1}{2}n^{\frac{1}{2}}N^{-1}\sum_{i=n+1}^{n+N}\{(\widehat\btheta_{j,{\rm SL}}-\overline\btheta_j)\trans\frac{\partial m_{jk}(\bz^i;\overline\bfeta,\overline\btheta_j)}{\partial \btheta_j}+(\widehat\btheta_{k,SL}-\overline\btheta_k)\trans\frac{\partial {m}_{kj}(\bz^i;\overline\bfeta,\overline\btheta_k)}{\partial \btheta_k}\},
    \end{aligned}
\end{equation*}
which can be written as 
\begin{equation*}
    \begin{aligned}
        A_2=&\frac{1}{2}n^{\frac{1}{2}}(\widehat\btheta_{j,{\rm SL}}-\overline\btheta_j)\trans\{n^{-1}\sum_{i=1}^n\frac{\partial m_{jk}(\bz^i;\overline\bfeta,\overline\btheta_j)}{\partial \btheta_j}-N^{-1}\sum_{i=n+1}^{n+N}\frac{\partial m_{jk}(\bz^i;\overline\bfeta,\overline\btheta_j)}{\partial\btheta_j}\}\\
        &+\frac{1}{2}n^{\frac{1}{2}}(\widehat\btheta_{k,SL}-\overline\btheta_k)\trans\{n^{-1}\sum_{i=1}^n\frac{\partial m_{kj}(\bz^i;\overline\bfeta,\overline\btheta_k)}{\partial \btheta_k}-N^{-1}\sum_{i=n+1}^{n+N}\frac{\partial m_{kj}(\bz^i;\overline\bfeta,\overline\btheta_k)}{\partial \btheta_k}\}.
    \end{aligned}
\end{equation*}

Since $n^{\frac{1}{2}}(\widehat\btheta_{j,{\rm SL}}-\overline\btheta_j)=O_p(1)$ and have
\[
n^{-1}\sum_{i=1}^n\frac{\partial m_{jk}(\bz^i;\overline\bfeta,\overline\btheta_j)}{\partial \btheta_j}={\rm E}\{\frac{\partial m_{jk}(\bz^i;\overline\bfeta,\overline\btheta_j)}{\partial \btheta_j}\}+o_p(1)
\]
by the weak law of large numbers, it follows that $A_2=O_p(1)\cdot o_p(1)=o_p(1).$
Then A equals
\begin{equation*}
\begin{aligned}
    A=A_1-A_2=&\frac{1}{2}n^{-\frac{1}{2}}\sum_{i=1}^n\{s_{jk}(\by_{\bw}^i;\overline\btheta_j)+s_{kj}(\by_{\bw}^i;\overline\btheta_k)-m_{jk}(\bz^i;\overline\bfeta,\overline\btheta_j)-m_{kj}(\bz^i;\overline\bfeta,\overline\btheta_k)\}\\
    &+\frac{1}{2}n^{\frac{1}{2}}N^{-1}\sum_{i=n+1}^{n+N}\{m_{jk}(\bz^i;\overline\bfeta,\overline\btheta_j)+m_{kj}(\bz^i;\overline\bfeta,\overline\btheta_k)\}+o_p(1) . 
\end{aligned}
\end{equation*}

We next analysis B, which can be written as 
\begin{equation*}
\begin{aligned}
B=&n^{\frac{1}{2}}(\widehat\bfeta-\overline\bfeta)\trans[\frac{1}{2}n^{-1}\sum_{i=1}^n\{\frac{\partial \widehat{m}_{jk}(\bz^i;\bfeta^{*1},\widehat\btheta_{j,{\rm SL}})}{\partial \bfeta}+\frac{\partial \widehat{m}_{kj}(\bz^i;\bfeta^{*2},\widehat\btheta_{k,SL})}{\partial \bfeta}\}\\
&-\frac{1}{2}N^{-1}\sum_{i=n+1}^{n+N}\{\frac{\partial \widehat{m}_{jk}(\bz^i;\bfeta^{*3},\widehat\btheta_{j,{\rm SL}})}{\partial \bfeta}+\frac{\partial \widehat{m}_{kj}(\bz^i;\bfeta^{*4},\widehat\btheta_{k,SL})}{\partial \bfeta}\}],
\end{aligned}
\end{equation*}
where $\bfeta^{*i},i=1,2,3,4$ are vectors between $\widehat\bfeta$ and $\overline\bfeta$.
Since $\widehat\bfeta\xrightarrow{p}\overline\bfeta$,
we have $\bfeta^*\xrightarrow[]{p}\overline\bfeta$. Also we have $\widehat\Sigma_{\widehat\btheta_{j,{\rm SL}}}\xrightarrow{p}\Sigma_{\overline\btheta_j}
$ since $\widehat\btheta_{j,{\rm SL}}\xrightarrow{p}\overline\btheta_j$ by the weak law of large numbers. It follows that  
\begin{equation*}
\frac{\partial \widehat{m}_{jk}(\bz;\bfeta^*,\widehat\btheta_{j,{\rm SL}})}{\partial \bfeta}=\frac{\partial m_{jk}(\bz;\overline\bfeta,\overline\btheta_j)}{\partial \bfeta}+o_p(1),
\end{equation*}
therefore, 
\begin{equation*}
    \begin{aligned}
    B=&n^{\frac{1}{2}}(\widehat\bfeta-\overline\bfeta)\trans\{\frac{1}{2}n^{-1}\sum_{i=1}^n\frac{\partial m_{jk}(\bz^i;\overline\bfeta,\overline\btheta_j)}{\partial \bfeta}-\frac{1}{2}N^{-1}\sum_{i=n+1}^{n+N}\frac{\partial m_{jk}(\bz^i;\overline\bfeta,\overline\btheta_j)}{\partial \bfeta}\\
    &+\frac{1}{2}n^{-1}\sum_{i=1}^n\frac{\partial m_{kj}(\bz^i;\overline\bfeta,\overline\btheta_k)}{\partial \bfeta}-\frac{1}{2}N^{-1}\sum_{i=n+1}^{n+N}\frac{\partial m_{kj}(\bz^i;\overline\bfeta,\overline\btheta_k)}{\partial \bfeta}+o_p(1)\}
    \end{aligned}
\end{equation*}
Same as in the process of part A, we have 
\begin{equation*}
    n^{-1}\sum_{i=1}^n\frac{\partial m_{jk}(\bz^i;\overline\bfeta,\overline\btheta_j)}{\partial \bfeta}={\rm E}\{\frac{\partial m_{jk}(\bz;\overline\bfeta,\overline\btheta_j)}{\partial \bfeta}\}+o_p(1),
\end{equation*}
it follows that 
\begin{equation*}
    \begin{aligned}
         B=&\frac{1}{2}n^{\frac{1}{2}}(\widehat\bfeta-\overline\bfeta)\trans[{\rm E}\{\frac{\partial m_{jk}(\bz;\overline\bfeta,\overline\btheta_j)}{\partial \btheta_j}\}-{\rm E}\{\frac{\partial m_{jk}(\bz;\overline\bfeta,\overline\btheta_j)}{\partial \btheta_j}\}\\
         &+{\rm E}\{\frac{\partial m_{kj}(\bz;\overline\bfeta,\overline\btheta_k)}{\partial \btheta_k}\}-{\rm E}\{\frac{\partial m_{kj}(\bz;\overline\bfeta,\overline\btheta_k)}{\partial \btheta_k}\}+o_p(1)]\\
         =&n^{\frac{1}{2}}(\widehat\bfeta-\overline\bfeta)\trans \cdot o_p(1)
    \end{aligned}
\end{equation*}
Then we have $B=O_p(1)\cdot o_p(1)=o_p(1).$

Above all, we have 
\begin{equation*}
    \begin{aligned}
        n^{\frac{1}{2}}(\widehat\theta_{jk,{\rm SCISS}}-\overline\theta_{jk})=&A-B+o_p(1)\\
        =&\frac{1}{2}n^{-\frac{1}{2}}\sum_{i=1}^n\{s_{jk}(\by_{\bw}^i;\overline\btheta_j)+s_{kj}(\by_{\bw}^i;\overline\btheta_k)-m_{jk}(\bz^i;\overline\bfeta,\overline\btheta_j)-m_{kj}(\bz^i;\overline\bfeta,\overline\btheta_k)\}\\
        &+\frac{1}{2}n^{\frac{1}{2}}N^{-1}\sum_{i=n+1}^{n+N}\{m_{jk}(\bz^i;\overline\bfeta,\overline\btheta_j)+m_{kj}(\bz^i;\overline\bfeta,\overline\btheta_k)\}+o_p(1)
    \end{aligned}
\end{equation*}

By the classic central limit theorem, we have
\begin{equation*}
    \begin{aligned}
         n^{\frac{1}{2}}(\widehat\theta_{jk,{\rm SCISS}}-\overline\theta_{jk})\xrightarrow{}N\left(0,\Omega_{jk, {\rm SCISS}}+ \frac{\zeta}{4} {\rm E} \{m_{jk}(\bz;\overline\bfeta,\overline\btheta_j)+m_{kj}(\bz;\overline\bfeta,\overline\btheta_k)\}^2\right)
    \end{aligned}
\end{equation*}
where $\Omega_{jk,{\rm SCISS}}=\frac{1}{4}{\rm E}\{s_{jk}(\by_{\bw};\overline\btheta_j)+s_{kj}(\by_{\bw};\overline\btheta_k)-m_{jk}(\bz;\overline\bfeta,\overline\btheta_j)-m_{kj}(\bz;\overline\bfeta,\overline\btheta_k)\}^2$ and $\zeta=\frac{n}{N}.$
By the underlying assumption that $\frac{n}{N}=o(1)$, the second term of the variance of $A$ equals $o_p(1)$. It follows that 
\begin{equation*}
    n^{\frac{1}{2}}(\widehat\theta_{jk,{\rm SCISS}}-\overline\theta_{jk})\xrightarrow{}N(0,\Omega_{jk,{\rm SCISS}})
\end{equation*}

Thus $n^{\frac{1}{2}}(\widehat{\theta}_{jk,{\rm SCISS}}-\overline\theta_{jk})$ converges to $N(0,\Omega_{jk,{\rm SCISS}})$ in distribution.

Notice that we have $\Omega_{jk,{\rm SL}}=\frac{1}{4}{\rm E}\{(s_{jk}(\by_{\bw};\overline\btheta_j)+s_{kj}(\by_{\bw};\overline\btheta_k)\}^2$ in Lemma \ref{lemma:SL} and $\Omega_{jk,{\rm SCISS}}=\frac{1}{4}{\rm E}\{m_{kj}(\bz;\overline\bfeta,\overline\btheta_k)+m_{jk}(\bz;\overline\bfeta,\overline\btheta_j)-s_{jk}(\by_{\bw};\overline\btheta_j)-s_{kj}(\by_{\bw};\overline\btheta_k)\}^2$ proved above. When the assumption model is correctly specified, $m_{jk}(\bz;\overline\bfeta,\overline\btheta_j)$ is the true conditional mean of $s_{jk}(\by;\overline\btheta_j)$ on $\bz$, combined with the well-known theorem that the variance of ${\rm E}(\mathbf{X}_2|\mathbf{X}_1)$ is no larger than the variance of $\mathbf{X}_2$, we have that $\Omega_{jk,{\rm SL}}-\Omega_{jk,{\rm SCISS}}\geq 0.$  

\section{Supplementary details for simulation}
\subsection{Procedure of Obtaining Intrinsic Efficient Estimator}\label{sec:procedure_intri}
In Section \ref{sec:intrinsic}, we introduce the intrinsic efficient SS estimation to ensure the variance of our improved SCISS estimator, $\widehat\theta_{jk,{\rm INTR}}$, is less than $\widehat\btheta_{j,{\rm SL}}$. 
Unfortunately, $\widehat\sigma_j^2(\bfeta,\e)$ is a non-convex function of $\bfeta$. So to ensure the stability, our strategy is to use a gradient method to update $\bfeta$ until $\widehat\sigma_j^2(\bfeta,\e)$ stops to descend. In specific,we let $\widehat\bfeta_{\rm INTR}^{(1)}=\widehat\bfeta_{\rm INTR}$ perform
\[
    \widehat\bfeta_{\rm INTR}^{(t)}=\widehat\bfeta_{\rm INTR}^{(t-1)}+\left[\frac{\partial ^2 \widehat\sigma_j^2(\bfeta,\e)}{\partial \bfeta\partial\bfeta\trans}\right]^{-1}\frac{\partial \widehat\sigma_j^2(\bfeta,\e)}{\partial \bfeta}\bigg|_{\bfeta=\widehat\bfeta_{\rm INTR}^{(t-1)}},
\]
where $\widehat\sigma_j^2(\bfeta,\e)=n^{-1}\sum_{i=1}^n [\e\trans\{\widehat\S_j(\by_{\bw}^i;\widehat\btheta_{j,{\rm SL}})-\widehat\m_j(\bz^i;\bfeta,\widehat\btheta_{j,{\rm SL}})\}]^2,$
for $t=1,2,\ldots,T,$ until $\widehat\sigma_j^2(\widehat\bfeta_{\rm INTR}^{(t-1)},\e)<\widehat\sigma_j^2(\widehat\bfeta_{\rm INTR}^{(t)},\e)$ or $t$ reaches the bound of iteration $T$. 

In practice, to facilitate the numerical calculation, we apply a method as to expand $\widehat\m_j(\bz;\bfeta,\widehat\btheta_{j,{\rm SL}})$ as 
\[
\widehat\m_j^*(\bz;\bfeta,\widehat\btheta_{j,{\rm SL}})=\widehat\m_j(\bz;\widehat\bfeta^{(t-1)},\widehat\btheta_{j,{\rm SL}})+(\bfeta-\widehat\bfeta^{(t-1)})\trans\frac{\partial \widehat\m_j(\bz;\widehat\bfeta^{(t-1)},\widehat\btheta_{j,{\rm SL}})}{\partial \bfeta}
\]
in our simulation study and the simulation result of intrinsic Efficient SS estimation is provided in Appendix \ref{sec:intrinsic simulation}.

\subsection{Simulation of Intrinsic Efficient SS estimation}\label{sec:intrinsic simulation}
We conducted the additional simulation studies to evaluate the performance of the Intrinsic Efficient SS estimation proposed in Section \ref{sec:misspecification}. We let $q=3$ and first generated $\by$ by Ising parameters
\begin{center}
$\btheta=\begin{bmatrix}-2&0&1\\0&-1.5&1\\1&1&-1\end{bmatrix}$.
\end{center}
Then we let $p=q=3$ and fix $p=q=3$, $n=500$, $N=7500$. 
To generate $\bx$, we let $P(x_k=1| y_k=1)=0.6$, $P(x_k=0| y_k=1)=0.4$ and set $x_k=0$ when $y_k=0$. The inspiration of this setting comes from the anchor-positive surrogate variable \citep{zhang2020maximum}.
Among all the $3\times3$ estimates, we target at $\theta_{12}$ and exert Intrinsic Efficient SS estimation procedure on its dimension. The first step is to obtain the SCISS-Aug estimator, then we use the iteration method described in Section \ref{sec:procedure_intri}, using Newton-Raphson on $\widehat\sigma_j^2(\bfeta,\e)$ to update the $\bgamma$ from its original estimate until the empirical $\widehat\sigma_j^2(\bfeta,\e)$ does not decrease anymore. We briefly name this strategy including additional steps as SCISS-Aug-intri. To ensure the algorithm is terminable, the bound of iterations is set as $4$. The result shows that the relative efficiency between SCISS-Aug-intri and SCISS-Aug is {$\mathbf{1.25}$}, confirming a moderately better performance than SCISS-Aug. 
Notice that the above setting has a relatively severe misspecification of the imputation model. In most cases, including the main setting introduced in Section \ref{sec:simulation}, the Intrinsic estimator performs close to the basic SCISS result. For example, the iteration on $\bfeta$ does not happen therefore no update version of $\widehat\theta_{jk,{\rm SCISS}}$ can be built, or the update does happen but the change is so small that the relative efficiency is close to 1. It reminds us that it is not necessary to add Intrinsic Efficient SS estimation in practice.

\subsection{Asymptotic Properties of $\widehat\btheta_{\rm INTR}$}\label{sec:proof_intri}

When we are discussing the properties of intrinsic SS estimator in this section, we assume that $\widehat\bfeta_{\rm INTR}$ is the global vector such that $\widehat\sigma^2_j(\bfeta,\e),j=1,\ldots,q$ reaches its minimization. 

Since $\widehat\bfeta_{\rm INTR}$ satisfied equation (\ref{equ:etaintri_E}) and we have the underlying assumption of global property above, we have $\widehat\bfeta_{\rm INTR}$ is the solution to
\[
\widehat{C}=\frac{2}{n}\sum_{i=1}^n\frac{\partial\{\e\trans\widehat\m_j(\bz^i;\bfeta,\widehat\btheta_{j,{\rm SL}})\}}{\partial\bfeta}[\e\trans\{\widehat\S_j(\by^i;\widehat\btheta_{j,{\rm SL}})-\widehat\m_j(\bz^i;\bfeta,\widehat\btheta_{j,{\rm SL}})\}]=0,
\]
for each $j$.
Under condition \ref{cond:intri}, we have the Taylor expansion of $n^{\frac{1}{2}}(\widehat\bfeta_{\rm INTR}-\overline\bfeta_{\rm INTR})$ as
\[
    n^{\frac{1}{2}}(\widehat\bfeta_{\rm INTR}-\overline\bfeta_{\rm INTR})=D^{-1}\sum_{i=1}^n C+o_p(1),
\]
where 
\[D={\rm E}\bigg(\frac{\partial(\e\trans\m_j(\bz;\overline\bfeta,\overline\btheta_j))}{\partial \bfeta}\cdot\frac{\partial(\e\trans\m_j(\bz;\overline\bfeta,\overline\btheta_j))}{\partial\bfeta\trans}-[\e\trans\{\S_j(\by_{\bw};\overline\btheta_j)-\m_j(\bz;\overline\bfeta,\overline\btheta_j)\}]\cdot\frac{\partial^2\{\e\trans\m_j(\bz;\overline\bfeta,\overline\btheta_j)\}}{\partial\bfeta\partial\bfeta\trans}\bigg)
\]
and 
\[C={\rm E}\bigg(\frac{\partial\{\e\trans\m_j(\bz^i;\bfeta,\overline\btheta_j)\}}{\partial\bfeta}[\e\trans\{\S_j(\by_{\bw}^i;\overline\btheta_j)-\m_j(\bz;\bfeta,\overline\btheta_j)\}\bigg).\] 

Therefore we have $n^{\frac{1}{2}}(\widehat\bfeta_{\rm INTR}-\overline\bfeta_{\rm INTR})$ converges to a Gaussian distribution, then $n^{\frac{1}{2}}(\widehat\bfeta_{\rm INTR}-\overline\bfeta_{\rm INTR})=O_p(1)$. Similar to the proof in Appendix \ref{sec:proof_SSL}, we can derive the expansion for $n^{\frac{1}{2}}(\widehat\theta_{jk,{\rm INTR}}-\overline\theta_{jk})$ as 
\begin{equation*}
\begin{aligned}
    n^{\frac{1}{2}}(\widehat\theta_{jk,{\rm INTR}}-\overline\theta_{jk})=&\frac{1}{2}n^{-\frac{1}{2}}\sum_{i=1}^{n}\{s_{jk}(\by_{\bw}^i;\overline\btheta_j)+ s_{kj}(\by_{\bw}^i;\overline\btheta_k)- m_{jk}(\bz^i;\overline\bfeta_{\rm INTR},\overline\btheta_j)- m_{kj}(\bz^i;\overline\bfeta_{\rm INTR},\overline\btheta_k)\}\\
    &+\frac{1}{2}n^{\frac{1}{2}}N^{-1}\sum_{i=n+1}^{n+N}\{ m_{jk}(\bz^i;\overline\bfeta_{\rm INTR},\overline\btheta_j)+ m_{kj}(\bz^i;\overline\bfeta_{\rm INTR},\overline\btheta_k)\}+o_p(1),
\end{aligned}
\end{equation*}
also
\begin{equation*}
\begin{aligned}
    n^{\frac{1}{2}}(\widecheck\btheta_{j,{\rm INTR}}-\overline\btheta_j)=&n^{-\frac{1}{2}}\sum_{i=1}^{n}\{\S_j(\by_{\bw}^i;\overline\btheta_j)-\m_j(\bz^i;\overline\bfeta_{\rm INTR},\overline\btheta_j)\}+n^{\frac{1}{2}}N^{-1}\sum_{i=n+1}^{n+N}\m_j(\bz^i;\overline\bfeta_{\rm INTR},\overline\btheta_j)+o_p(1).
\end{aligned}
\end{equation*}

From Theorem $\ref{thm:SSL}$, we can learn that when the imputation model is correctly specified, that is, $\overline\bfeta_{\rm INTR}=\overline\bfeta$, it follows that $n^{\frac{1}{2}}(\widehat\theta_{jk,{\rm INTR}}-\overline\theta_{jk})$ is asymptotic equivalent to $n^{\frac{1}{2}}(\widehat\theta_{jk,{\rm SCISS}}-\overline\theta_{jk}).$ When the imputation model is mis-specified, we are able to obtain an estimator for $\theta_{jk,{\rm INTR}}$ by traversing $\e\in\{0,1\}^q$ for each $j$ thus obtain $\widehat\theta_{jk,{\rm INTR}}$ that owns smallest variance since $\overline\bfeta_{\rm INTR}$
is the minimum point of the function of asymptotic function. This completes the proof of Theorem \ref{thm:intri}.

\newpage
\subsection{Complete Result of Simulation}\label{sec:complete simulation result}
\begin{table}[!ht]
\caption{The complete result of simulation that includes the numerical analysis of diagonal entry. The formula of measuring the value is consistent with the simulation section\ref{sec:simulation}.}
\begin{center}
( gaussian none )\vspace{1em}\\
\begin{tabular}{|c|c|c|c|c|c|c|c|c|c|c|c|c|c|}\hline
{}&\multicolumn{2}{c|}{SL}& \multicolumn{3}{c|}{DR}&\multicolumn{4}{c|}{SCISS-Aug}&\multicolumn{4}{c|}{SCISS-PoS}\\\hline 
{}&Bias&SE&Bias&SE&RE&Bias&SE&RE&CP&Bias&SE&RE&CP\\\hline
{$\theta_{11}$}& 0.005 & 0.246 & 0.005 & 0.216 & 1.31 & 0.007 & 0.220 & 1.25 & 0.96 & -0.007 & \pmb{0.165} & \pmb{2.23} & 0.94\\
{$\theta_{12}$}& 0.002 & 0.294 & 0.002 & 0.295 & 1.00 & -0.005 & 0.242 & 1.49 & 0.96 & -0.014 & \pmb{0.208} & \pmb{2.01} & 0.95\\
{$\theta_{13}$}& -0.013 & 0.290 & -0.014 & 0.291 & 0.99 & -0.005 & 0.276 & 1.10 & 0.95 & 0.031 & \pmb{0.219} & \pmb{1.74} & 0.94\\
{$\theta_{22}$}& -0.011 & 0.280 & -0.004 & 0.250 & 1.26 & 0.002 & 0.244 & 1.31 & 0.96 & 0.002 & \pmb{0.183} & \pmb{2.33} & 0.95\\
{$\theta_{23}$}& 0.003 & 0.289 & 0.002 & 0.291 & 0.99 & 0.007 & 0.288 & 1.01 & 0.96 & 0.006 & \pmb{0.212} & \pmb{1.86} & 0.95\\
{$\theta_{33}$}& 0.010 & 0.240 & 0.006 & 0.215 & 1.24 & 0.009 & 0.203 & 1.41 & 0.95 & -0.022 & \pmb{0.179} & \pmb{1.81} & 0.92\\\hline
\end{tabular}
\end{center}
\end{table}

\begin{table}[!ht]
\begin{center}
( gaussian weak )\vspace{1em}\\
\begin{tabular}{|c|c|c|c|c|c|c|c|c|c|c|c|c|c|}\hline
{}&\multicolumn{2}{c|}{SL}& \multicolumn{3}{c|}{DR}&\multicolumn{4}{c|}{SCISS-Aug}&\multicolumn{4}{c|}{SCISS-PoS}\\\hline 
{}&Bias&SE&Bias&SE&RE&Bias&SE&RE&CP&Bias&SE&RE&CP\\\hline
$\theta_{11}$& 0.005 & 0.246 & 0.006 & 0.224 & 1.21 & 0.007 & 0.204 & 1.46 & 0.95 & 0.001 & \pmb{0.198} & \pmb{1.56} & 0.94\\
$\theta_{12}$& 0.002 & 0.294 & 0.003 & 0.295 & 1.00 & -0.003 & 0.236 & 1.55 & 0.95 & -0.013 & \pmb{0.237} & \pmb{1.55} & 0.95\\
$\theta_{13}$& -0.013 & 0.290 & -0.015 & 0.289 & 1.00 & -0.004 & 0.243 & 1.42 & 0.95 & 0.033 & \pmb{0.233} & \pmb{1.55} & 0.95\\
$\theta_{22}$& -0.011 & 0.280 & -0.004 & 0.257 & 1.18 & 0.001 & 0.231 & 1.46 & 0.96 & 0.006 & \pmb{0.216} & \pmb{1.69} & 0.94\\
$\theta_{23}$& 0.003 & 0.289 & 0.002 & 0.291 & 0.99 & 0.009 & 0.262 & 1.22 & 0.96 & 0.004 & \pmb{0.223} & \pmb{1.68} & 0.95\\
$\theta_{33}$& 0.010 & 0.240 & 0.007 & 0.224&1.15& 0.006 & \pmb{0.198} & \pmb{1.48} & 0.95 & -0.037 & 0.207 &1.35 & 0.93 \\\hline
\end{tabular}
\end{center}
\end{table}

\begin{table}[!ht]
\begin{center}
( gaussian moderate )\vspace{1em}\\
\begin{tabular}{|c|c|c|c|c|c|c|c|c|c|c|c|c|c|}\hline
{}&\multicolumn{2}{c|}{SL}& \multicolumn{3}{c|}{DR}&\multicolumn{4}{c|}{SCISS-Aug}&\multicolumn{4}{c|}{SCISS-PoS}\\\hline 
{}&Bias&SE&Bias&SE&RE&Bias&SE&RE&CP&Bias&SE&RE&CP\\\hline
$\theta_{11}$& 0.005 & 0.246 & 0.006 & 0.239 & 1.07 & 0.002 & \pmb{0.203} & \pmb{1.47} & 0.95 & 0.001 & 0.287 & 0.74 & 0.93\\
$\theta_{12}$& 0.002 & 0.294 & 0.006 & 0.292 & 1.02 & 0.003 & \pmb{0.225} & \pmb{1.71} & 0.94 & -0.021 & 0.270 & 1.19 & 0.94\\
$\theta_{13}$& -0.013 & 0.290 & -0.015 & 0.280 & 1.07 & 0.001 & \pmb{0.223} & \pmb{1.68} & 0.95 & 0.032 & 0.271 & 1.15 & 0.94 \\
$\theta_{22}$& -0.011 & 0.280 & -0.005 & 0.266 & 1.11 & -0.000 & \pmb{0.220} & \pmb{1.62} & 0.96 & -0.041 & 0.262 & 1.14 & 0.92\\
$\theta_{23}$& 0.003 & 0.289 & 0.003 & 0.288 & 1.01 & 0.010 & \pmb{0.235} & \pmb{1.51} & 0.96 & -0.001 & 0.269 & 1.16 & 0.94\\
$\theta_{33}$& 0.010 & 0.240 & 0.009 & 0.234 & 1.05 & 0.001 & \pmb{0.202} & \pmb{1.42} & 0.96 & 0.011 & 0.292 & 0.68 & 0.92\\\hline
\end{tabular}
\end{center}
\end{table}

\begin{table}[!ht]
\begin{center}
( poisson none )\vspace{1em}\\
\begin{tabular}{|c|c|c|c|c|c|c|c|c|c|c|c|c|c|}\hline
{}&\multicolumn{2}{c|}{SL}& \multicolumn{3}{c|}{DR}&\multicolumn{4}{c|}{SCISS-Aug}&\multicolumn{4}{c|}{SCISS-PoS}\\\hline 
{}&Bias&SE&Bias&SE&RE&Bias&SE&RE&CP&Bias&SE&RE&CP\\\hline
$\theta_{11}$& 0.005 & 0.246 & 0.008 & 0.216 & 1.30 & -0.001 & 0.220 & 1.25 & 0.96 & 0.002 & \pmb{0.103} & \pmb{5.75} & 0.92\\
$\theta_{12}$& 0.002 & 0.294 & 0.002 & 0.293 & 1.01 & 0.004 & 0.228 & 1.67 & 0.97 & 0.003 & \pmb{0.129} & \pmb{5.18} & 0.94\\
$\theta_{13}$& -0.013 & 0.290 & -0.014 & 0.291 & 0.99 & -0.007 & 0.279 & 1.08 & 0.96 & 0.005 & \pmb{0.134} & \pmb{4.66} & 0.93\\
$\theta_{22}$& -0.011 & 0.280 & -0.011 & 0.243 & 1.33 & -0.007 & 0.235 & 1.42 & 0.96 & 0.002 & \pmb{0.113} & \pmb{6.13} & 0.95\\
$\theta_{23}$& 0.003 & 0.289 & 0.004 & 0.289 & 1.00 & 0.004 & 0.287 & 1.01 & 0.96 & -0.005 & \pmb{0.127}& \pmb{5.21} & 0.95\\
$\theta_{33}$& 0.010 & 0.240 & 0.004 & 0.211 & 1.29 & 0.012 & 0.191 & 1.58 & 0.95 & -0.003 & \pmb{0.102} & \pmb{5.57} & 0.92\\\hline
\end{tabular}
\end{center}
\end{table}

\begin{table}[!ht]
\begin{center}
( poisson weak )\vspace{1em}\\
\begin{tabular}{|c|c|c|c|c|c|c|c|c|c|c|c|c|c|}\hline
{}&\multicolumn{2}{c|}{SL}& \multicolumn{3}{c|}{DR}&\multicolumn{4}{c|}{SCISS-Aug}&\multicolumn{4}{c|}{SCISS-PoS}\\\hline 
{}&Bias&SE&Bias&SE&RE&Bias&SE&RE&CP&Bias&SE&RE&CP\\\hline
$\theta_{11}$& 0.005 & 0.246 & 0.006 & 0.218 & 1.28 & -0.003 & 0.207 & 1.41 & 0.96 & 0.002 & \pmb{0.118} &\pmb{4.34} & 0.92\\
$\theta_{12}$& 0.002 & 0.294 & 0.005 & 0.296 & 0.99 & 0.003 & 0.222 & 1.76 & 0.96 & 0.003 & \pmb{0.141} & \pmb{4.38} & 0.94\\
$\theta_{13}$& -0.013 & 0.290 & -0.015 & 0.290 & 1.00 & -0.005 & 0.254 & 1.30 & 0.95 & 0.005 & \pmb{0.132} & \pmb{4.79} & 0.96\\
$\theta_{22}$& -0.011 & 0.280 & -0.006 & 0.247 & 1.29 & -0.009 & 0.229 & 1.50 & 0.95 & 0.004 & \pmb{0.131} & \pmb{4.59} & 0.94\\
$\theta_{23}$& 0.003 & 0.289 & 0.004 & 0.290 & 1.00 & 0.011 & 0.280 & 1.07 & 0.97 & -0.010 & \pmb{0.144} & \pmb{4.04} & 0.94\\
$\theta_{33}$& 0.010 & 0.240 & 0.009 & 0.214 & 1.26 & 0.011 & 0.189 & 1.61 & 0.95 & 0.004 & \pmb{0.115} & \pmb{4.39} & 0.94\\\hline
\end{tabular}
\end{center}
\end{table}

\begin{table}[!ht]
\begin{center}
( poisson moderate )\vspace{1em}
\\
\begin{tabular}{|c|c|c|c|c|c|c|c|c|c|c|c|c|c|}\hline
{}&\multicolumn{2}{c|}{SL}& \multicolumn{3}{c|}{DR}&\multicolumn{4}{c|}{SCISS-Aug}&\multicolumn{4}{c|}{SCISS-PoS}\\\hline 
{}&Bias&SE&Bias&SE&RE&Bias&SE&RE&CP&Bias&SE&RE&CP\\\hline
$\theta_{11}$& 0.005 & 0.246 & 0.005 & 0.231 & 1.14 & -0.002 & 0.197 & 1.57 & 0.96 & -0.010 & \pmb{0.178} & \pmb{1.92} & 0.93\\
$\theta_{12}$& 0.002 & 0.294 & 0.007 & 0.296 & 0.99 & 0.005 & 0.209 & 1.99 & 0.96 & 0.006 & \pmb{0.173} &\pmb{2.89} & 0.94\\
$\theta_{13}$& -0.013 & 0.290 & -0.015 & 0.284 & 1.04 & -0.004 & 0.226 & 1.65 & 0.95 & 0.017 & \pmb{0.176} & \pmb{2.70} & 0.96\\
$\theta_{22}$& -0.011 & 0.280 & -0.005 & 0.255 & 1.20 & -0.007 & 0.214 & 1.71 & 0.96 & -0.009 & \pmb{0.186} & \pmb{2.27} & 0.93\\
$\theta_{23}$& 0.003 & 0.289 & 0.006 & 0.289 & 1.00 & 0.010 & 0.254 & 1.30 & 0.96 & -0.012 & \pmb{0.187} & \pmb{2.39} & 0.92\\
$\theta_{33}$& 0.010 & 0.240 & 0.010 & 0.224 & 1.15 & 0.009 & 0.189 & 1.62 & 0.96 & 0.003 & \pmb{0.174} & \pmb{1.91} & 0.95 \\\hline
\end{tabular}
\end{center}
\end{table}

\newpage
\subsection{Procedure of Density-ratio}\label{sec:dr}
In this section we explain the detail of implementing the density-ratio method in simulation. We use the same procedure in real data analysis Section \ref{sec:realdata}.
We apply the alternative semi-supervised learning method density-ratio estimation proposed by \cite{2012Semi} on simulation studies and example of MIMIC-III and compare the result with the SCISS method we proposed. Density-ratio method introduces the weighted MLE with the weight function that can be estimated by the labeled and unlabeled data. When we implement this method, we construct a new response $\bX_{dr}$ with length $n+N$ and let the first n value (refers to the sample size of labeled data) to be 0 and the next $N$ value to be 1 (refers to the sample size of unlabeled data). At the same time, we construct a new $(n+N)\times p$ variable $\bY_{dr}$, then replicate $\bx$ of labeled data into the first n array and $\bx$ of unlabeled data into the next N array of $\bY_{dr}$. We next obtain a regression coefficient $\widehat\bfeta_{dr}$ from the logistic regression of $\bY_{dr}$ on $\bX_{dr}$. Then we calculate the weight $\bw_{dr}(\bx)$, by the formula $w_{dr}(\bx)=\exp(\widehat\bfeta_{dr}\cdot\bx)$. Each $\bx$ of labeled data build its weight with this weight function. Next we normalize the weight for n labeled data as $\bw_{dr}=(w_{1,dr},\ldots,w_{n,dr})\trans$, where $w_{i,dr}=\exp(\widehat\bfeta_{dr}\cdot\bx^i)$. Then we solve the weighted MLE, $\widehat\btheta_{j,dr}=\{(\widehat\theta_{j1,dr})\trans,\ldots,(\widehat\btheta_{jj,dr})\trans,\ldots,(\widehat\theta_{jq,dr})\trans\}\trans$, as a solution of the following equation:
\[
\frac{1}{n}\sum_{i=1}^n \bw_{dr}[i]\cdot[\by^{i}_{\setminus j,\bw}\{y_j^i-g(\btheta_j\trans\by^{i}_{\setminus j,\bw})\}]=\bzero,
\]
and compose the first density ratio estimator as $\widecheck\btheta_{dr}=\{(\widehat\btheta_{1,dr})\trans,\ldots,(\widehat\btheta_{q,dr})\trans\}\trans$. In the same way, we apply the symmetrization on $\widecheck\btheta_{dr}$ as $\widehat\theta_{jk,dr}=\frac{1}{2}(\widecheck\theta_{jk,dr}+\widecheck\theta_{kj,dr})$ and obtain the final density ratio estimator as $\widehat\btheta_{dr}=(\widehat\btheta_{1,dr},\ldots,\widehat\btheta_{q,dr})\trans,$ where $\widehat\btheta_{j,dr}=(\widehat\theta_{j1,dr},
\ldots,\widehat\btheta_{jj,dr}\trans,\ldots,\widehat\theta_{jq,dr})\trans$ and $\widehat\btheta_{jj,dr}=\widecheck\btheta_{jj,dr}$ for $j=1,\ldots,q$.
\end{document}